\documentclass[acmsmall]{acmart}

\usepackage{xspace}
\usepackage{ifthen}
\usepackage{marginnote}
\usepackage{booktabs,tabularx,multirow}

\usepackage{microtype}

\usepackage{graphicx}
\usepackage{caption}
\usepackage{enumitem}
\usepackage{xcolor}

\usepackage{mathtools,amssymb,amsfonts,bm}
\usepackage{amsthm}

\usepackage{algorithm}
\usepackage{algorithmic}

\usepackage{titlecaps}

\usepackage{subcaption}

\newtheorem{definition}{Definition}
\newtheorem{example}{Example}

\makeatletter
\g@addto@macro\normalsize{%
  \setlength{\abovedisplayskip}{0pt}%
  \setlength{\abovedisplayshortskip}{0pt}%
  \setlength{\belowdisplayskip}{0pt}%
  \setlength{\belowdisplayshortskip}{0pt}%
}

\usepackage{array}
\usepackage{arydshln}
\makeatother

\usepackage{tikz}

\newcommand{\junhao}[1]{}
\newcommand{\daomin}[1]{}

\makeatletter
\def\thm@space@setup{%
  \thm@preskip=2pt
  \thm@postskip=2pt
}
\makeatother
\newcommand{\problem}{taxonomy maintenance in the wild\xspace}

\newcommand{\method}{\textsc{GIST}\xspace}

\newcommand{\methods}{\textsc{GIST}$^*$\xspace}

\newcommand{\wtob}{\ensuremath{\Phi_{\text{W2B}}}\xspace}
\newcommand{\btow}{\ensuremath{\Phi_{\text{B2W}}}\xspace}

\newcommand{\taxo}{\ensuremath{\mathcal{T}}\xspace}
\newcommand{\edge}{\ensuremath{\mathcal{E}}\xspace}
\newcommand{\node}{\ensuremath{\mathcal{V}}\xspace}
\newcommand{\repo}{\ensuremath{\mathcal{P}}\xspace}

\newcommand{\ctxabs}{\ensuremath{\textsc{abs}}} 

\newcommand{\nodefone}{\emph{Node Soft F1}\xspace}
\newcommand{\edgefone}{\emph{Edge Soft F1}\xspace}

\newcommand{\taxoadapt}{TaxoAdapt\xspace}
\newcommand{\taxoalign}{TaxoAlign\xspace}
\newcommand{\taxogen}{TaxoGen\xspace}
\newcommand{\hyperexpan}{HyperExpan\xspace}

\newcommand{\methodvar}[2]{$\text{GIST}_{\textsc{#1},\textsc{#2}}$\xspace}
\newcommand{\adaptvar}[2]{$\text{TaxoAdapt}_{\textsc{#1},\textsc{#2}}$\xspace}

\newcommand{\vol}{\operatorname{Vol}}

\newif\ifshowrevtag
\showrevtagfalse   

\definecolor{revCommon}{RGB}{231,76,60}    
\definecolor{revR1}{RGB}{230,126,34}       
\definecolor{revR2}{RGB}{52,152,219}       
\definecolor{revR3}{RGB}{142,68,173}       

\newif\ifshowrevtext
\showrevtextfalse  

\newif\ifshowreview
\showreviewfalse   

\AtBeginDocument{%
  }

\setcopyright{cc}
\setcctype{by}
\acmJournal{PACMMOD}
\acmYear{2026} \acmVolume{4} \acmNumber{4 (SIGMOD)} \acmArticle{289}
\acmMonth{9} \acmDOI{10.1145/3837127}

\begin{document}
\setlength{\abovecaptionskip}{0pt}
\setlength{\belowcaptionskip}{0pt}
\title[Taxonomy Maintenance in the Wild]{Taxonomy Maintenance in the Wild over Evolving Scholarly Data: Reliability, Efficiency, and Cost-Effectiveness}

\author{Daomin Ji}
\authornote{This work was done when Daomin Ji was a visiting student at The University of Queensland.}
\email{daomin.ji@student.rmit.edu.au}
\orcid{0009-0000-0037-3614}
\affiliation{%
  \institution{RMIT University}
  \city{Melbourne}
  \state{Victoria}
  \country{Australia}
}
\affiliation{%
  \institution{The University of Queensland}
  \city{Brisbane}
  \state{Queensland}
  \country{Australia}
}
\author{Hui Luo}
\email{huil@uow.edu.au}
\orcid{0000-0002-7299-031X}
\affiliation{%
  \institution{University of Wollongong}
  \city{Wollongong}
  \state{New South Wales}
  \country{Australia}
}
\author{Zhifeng Bao}
\authornote{Zhifeng Bao is the corresponding author.}
\email{zhifeng.bao@uq.edu.au}
\orcid{0000-0003-2477-381X}
\affiliation{%
  \institution{The University of Queensland}
  \city{Brisbane}
  \state{Queensland}
  \country{Australia}
}
\author{Junhao Gan}
\email{junhao.gan@unimelb.edu.au}
\orcid{0000-0001-9101-1503}
\affiliation{%
  \institution{The University of Melbourne}
  \city{Melbourne}
  \state{Victoria}
  \country{Australia}
}
\author{Zi Huang}
\email{huang@itee.uq.edu.au}
\orcid{0000-0002-9738-4949}
\affiliation{%
  \institution{The University of Queensland}
  \city{Brisbane}
  \state{Queensland}
  \country{Australia}
}

\renewcommand{\shortauthors}{Daomin Ji, Hui Luo, Zhifeng Bao, Junhao Gan, \& Zi Huang}

\begin{abstract}
The rapid growth of scientific publications makes scholarly taxonomies quickly obsolete. 
We study \emph{taxonomy maintenance in the wild}, a new problem that moves beyond static construction by continuously adapting taxonomies to evolving scholarly repositories, such as arXiv, for a given research topic.
We propose \method, a robust framework for maintaining evolving taxonomies.
Unlike purely LLM-centric approaches, \method grounds structure induction in expert-curated evidence by extracting partial hierarchies from the ``Related Work'' sections of papers.
It integrates these partial taxonomies into a unified global taxonomy in a \emph{geometric box-embedding space}, where box containment encodes the inductive bias of is-a relations.
To connect semantics with geometric structure, \method learns a bidirectional mapping between word embeddings and box embeddings.
For efficient incremental updates, \method uses novelty-aware coreset selection to update the model with representative historical signals and new evidence, avoiding costly full retraining.
To handle high-velocity paper streams under user-specific token budgets, \method further combines a hypothesized concept generator with a cost-effective evidence retrieval module.
Experiments on real-world arXiv datasets show that \method outperforms state-of-the-art baselines, improving Node F1 and Edge F1 by 11.0\% and 13.1\% over the strongest baseline while requiring only 9.6\% of its runtime and 12.7\% of its monetary cost.
\end{abstract}

\ccsdesc[500]{Information systems~Data structures}
\ccsdesc[500]{Information systems~Information integration}

\keywords{Taxonomy Maintenance, Scholarly Data Management, Semantic Index}


\received{17 January 2026}
\received[revised]{19 May 2026}
\received[accepted]{12 June 2026}
\maketitle

\section{Introduction}
\begin{figure}
    \centering
    \includegraphics[width=0.75\linewidth]{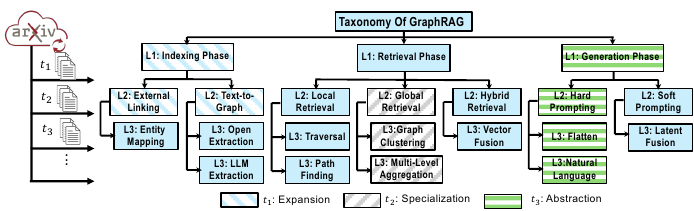}
    \caption{The rapid development of research topic \textit{GraphRAG}}
    \label{fig:problem}
\end{figure}

Taxonomies provide structured abstractions of domain knowledge by organizing concepts through hierarchical ``is-a'' relationships. They have long served as foundational infrastructure in data management, supporting tasks such as knowledge base construction~\cite{Suchanek2008YAGO,Auer2007DBpedia}, semantic indexing and query reformulation~\cite{Fontoura2008Relaxation,Ding2012OptimizingIndex}, Text-to-SQL translation~\cite{sen2020athena++,saha2016athena}, and product catalog management~\cite{Kanagal2012Supercharging}.

In this work, we focus on topic-centric taxonomies over scholarly data—fine-grained hierarchies tailored to a specific research topic—that organize knowledge from high-level problems and methodological families down to concrete design choices. 
A critical characteristic of scholarly taxonomies is their \textit{volatility}:  
Unlike taxonomies in general-purpose knowledge bases, whose conceptual hierarchies evolve relatively gradually, the scientific landscape for certain research topics evolves rapidly, as illustrated in Example~\ref{ex:motivating_example}. 
Consequently, existing taxonomies can quickly become obsolete, potentially impairing the key scholarly data management workflows:
(1) \textit{Reviewer Assignment}: Outdated classification schemes fail to capture emerging sub-fields, leading to  misclassified papers and mismatched reviewer assignments~\cite{acmccs2012}; 
(2) \textit{Information Retrieval}: Digital libraries relying on static indices, such as ACM Digital Library~\cite{acmdl} and IEEE Xplore~\cite{ieeexplore}, cannot effectively support semantic navigation for newly coined terminology, creating gaps between user queries and relevant literature;
(3) \textit{Conference Organization}~\cite{chilton2014frenzy,tanaka2002granulation}: conference programs fail to structurally reflect the shifting distribution of research interests, resulting in coarse-grained sessions.

\begin{example}
\label{ex:motivating_example}
Consider \textit{Graph Retrieval-Augmented Generation} (GraphRAG) in Fig.~\ref{fig:problem}, a rapidly evolving domain where researchers customize LLMs by organizing external knowledge as graphs, retrieving graph-structured evidence, and integrating it into generation. 
Since Microsoft introduced the initial GraphRAG method in 2024~\cite{edge2024local}, the field has expanded rapidly. 
As of January 2026, over 730 papers on arXiv leverage or discuss this paradigm. 
As the literature evolves, the taxonomy may undergo the following structural refinements at different timestamps:
(1) \textbf{expansion} (at $t_1$), which introduces a new node under an existing non-leaf category (e.g., adding \textit{External Linking} as a sub-category under \textit{Indexing Phase}, parallel to \textit{Text-to-Graph});
(2) \textbf{specialization} (at $t_2$), which splits a coarse concept into finer-grained variants (e.g., dividing \textit{Global Retrieval} into \textit{Graph Clustering} and \textit{Multi-Level Aggregation});
and (3) \textbf{abstraction} (at $t_3$), which adds an intermediate high-level parent to unify multiple existing paradigms (e.g., adding \textit{Hard Prompting} under the original parent \textit{Generation Phase} to unify \textit{Flatten} and \textit{Natural Language}).
\end{example}

This motivates the problem of  maintaining \emph{topic-centric} taxonomies over a large,  evolving scholarly repository (e.g., arXiv). 
Given a target research topic and an evolving scholarly repository that reflects the shifting research landscape, the goal is to automatically construct and incrementally adapt a taxonomy, ensuring it evolves alongside the underlying literature.

Most existing scholarly taxonomy construction methods~\cite{lahiri2025taxoalign,kargupta2025taxoadapt} are  designed for small, static corpora and cannot scale to large, dynamic settings due to three key limitations: (1) They heavily rely on LLMs not only to extract salient signals but also to reason about hierarchical structure. 
Consequently, they are vulnerable to hallucination: without reliable textual cues, LLM-inferred hierarchies often contain plausible yet unsupported concepts or relations.
(2) They lack native support for continual updates. 
With new arrivals, one must either re-infer the taxonomy from scratch over all accumulated documents, or build a delta taxonomy from the new batch and merge it via entity resolution~\cite{tu2023unicorn}, which may fix naming inconsistencies but often fails to resolve relational conflicts, causing drift (e.g., cycles or duplicated branches). As evidenced in Sec.~\ref{sec:main_results}, reconstruction from scratch increases monetary and time costs by up to 3.9$\times$, while naive incremental update degrades performance by 14\%--18\%. In Example~\ref{ex:motivating_example}, a single research topic typically can involve hundreds or thousands of relevant papers, making \emph{full-corpus processing} over all relevant papers infeasible.

To guide our method design, we identify three key desiderata. 
\textbf{D1: Reliable structure induction.}
It should infer taxonomic concepts and relations from reliable signals (rather than unconstrained LLM reasoning) to reduce hallucination. 
\textbf{D2: Efficient incremental updates.}
Given the continuous arrival of new papers, it should support efficient incremental updates that incorporate emerging knowledge without expensive recomputation. 
\textbf{D3: Cost-effective evidence acquisition.}
Under a fixed budget (e.g., token budget), it should selectively process the most informative documents to maximize evidence utility per unit cost.
To meet the above desiderata, we develop \method, a \underline{G}eometric-\underline{I}nferred and \underline{S}elf-refining \underline{T}axonomy Maintenance framework (Fig.~\ref{fig:workflow}). Our contributions include:
\begin{itemize}[leftmargin=*,noitemsep,topsep=0pt]

\item 
To satisfy \textit{D1}, \method derives taxonomic concepts and relations from expert-curated textual cues by using LLMs primarily as \textit{structure extractors}: it distills topic-relevant partial taxonomies from the \textit{Related Work} (or equivalent) sections of retrieved papers where authors explicitly organize and categorize prior work, and integrates these high-precision yet partial structures into a unified global taxonomy, mitigating hallucinated nodes and edges. 
This \emph{taxonomy integration} is formulated in a geometric box-embedding space, where box containment encodes inductive bias for ``is-a'' relations. 
To bridge semantic and geometric spaces, \method learns a \emph{self-supervised bidirectional word--box mapping model} from the extracted partial taxonomies and applies reliability-aware filtering to  retain high-confidence signals. (Sec.~\ref{sec:box_embedding}--\ref{sec:bidirectional_mappingd_model})

\item 
With this mapping model, \method avoids redundant reprocessing of historical content by transforming only novel concepts from newly extracted partial taxonomies into geometric boxes. To support \textit{D2} and continually adapt the mapping model, we introduce a novelty-aware coreset selection strategy. This strategy updates the bidirectional mapping model using only a representative subset of historical signals alongside new evidence, enabling efficient assimilation of emerging concepts without full retraining. (Sec.~\ref{sec:incremental_training})

\item
To support \textit{D3}, \method introduces two tightly coupled modules:
(1) Instead of using only the input topic keywords or existing taxonomy concepts as search signals---which can be overly broad and often retrieve redundant evidence—we propose a \textit{hypothesized concept generation} module that predicts likely emerging concepts for the next iteration (Sec.~\ref{sec:missing_box}). 
(2) Given a user-specified token budget, we develop a \textit{cost-effective evidence retrieval} module that aims to select an optimal query set of hypothesized concepts and retrieve (and process) only the most informative newly arrived papers to extract the partial taxonomies.

\item We  build an \emph{incremental semantic index} atop \method to enhance downstream paper search accuracy. This index is also maintained incrementally—mapping only newly arrived papers to the evolving taxonomy—thereby avoiding costly full repository rescanning to rebuild the index. (Sec.~\ref{sec:incremental_index})

\item Experiments on the real-world \textit{arXiv} scholarly corpus across diverse topics show that \method consistently constructs taxonomies with higher coverage and structural quality than the baselines. 
In particular, \method achieves average relative improvements of 11.0\% and 13.1\% in \nodefone and \edgefone over the strongest baseline, while requiring only 9.6\% of its runtime and 12.7\% of its monetary cost; the incremental update further delivers a 3.1$\times$ speedup over full retraining.

\end{itemize}

\section{Problem  and Solution Overview}
\label{sec:problem formulation and solution overview}

\begin{figure}
    \centering
    \includegraphics[width=0.7\linewidth]{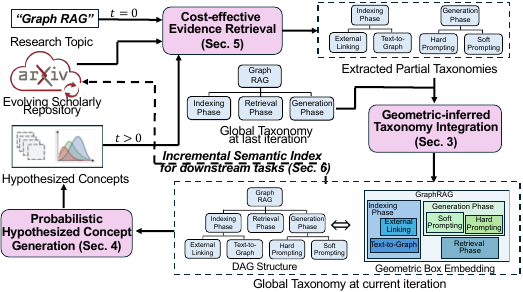}
    \caption{The workflow of the proposed \method}
    \label{fig:workflow}
\end{figure}

\subsection{Problem Formulation}\label{sec:problem definition}
\noindent \textbf{Concept and Taxonomy.} A scholarly taxonomy is a rooted directed acyclic graph (DAG) $\taxo =(\node,\edge)$, where each node $v\in \node$ is a \textit{concept}---a named semantic \emph{category} used to classify documents within the target topic, represented by a textual label. 
Each directed edge $(p,q)\in \edge$ encodes a hypernym--hyponym (parent--child) relation, meaning that $q$ is a sub-category of $p$. 
The taxonomy is rooted at a single concept $r$ representing the overarching research topic. 

\noindent \textbf{Scholarly Data Repository.} The scholarly data repository is an evolving set of documents, e.g., research papers, denoted at time $t$ as $\repo_t = \{d_1, d_2, \dots, d_{|\repo_t|}\}$. 
Each document $d_i$ comprises textual content (e.g., title, abstract, and full text).
Crucially, we adopt an \emph{open-world} setting (e.g., arXiv), where uncurated and continuously incoming documents yield a repository that mixes topic-relevant documents with off-topic or otherwise irrelevant documents.

The continuous evolution of the repository necessitates \emph{dynamic structural refinement} of the taxonomy:

\noindent \textbf{Structural Refinement.} We classify structural refinement into three types based on how a new concept relates to the current hierarchy:
\begin{itemize}[leftmargin=*,noitemsep,topsep=0pt]
\item \textbf{Expansion Refinement} occurs when a novel paradigm emerges at the same hierarchical level as existing categories under a common parent. It is realized by adding a new \textit{sibling node} to form a parallel (sub)-branch.
\item \textbf{Specialization Refinement} occurs when an existing leaf node becomes too coarse as the literature grows. It is realized by introducing finer-grained \textit{child nodes} to increase granularity.
\item \textbf{Abstraction Refinement} occurs when a set of sibling nodes naturally suggests a higher-level shared category. It is realized by inserting an \textit{intermediate node} to group these siblings.
\end{itemize}

\smallskip
Based on the preliminaries, we define our problem as follows:

\begin{definition}[\problem]
\label{def:problem}
Given an input research topic $r$ and an  evolving repository $\repo_t$, the problem of \problem is twofold:
\begin{itemize}[leftmargin=*,noitemsep,topsep=0pt]
\item \textbf{Initial Construction ($t=0$):} Construct an initial taxonomy $\taxo_0$ rooted at $r$ using the initial corpus snapshot $\repo_0$, establishing the baseline semantic structure of the domain.
\item \textbf{Continuous Evolution ($t>0$):} For each subsequent timestamp, refine $\taxo_t$ in light of newly arriving scientific documents $\Delta\repo_{t+1}=\repo_{t+1}\setminus \repo_t$ by applying structural refinements (e.g., expansion, specialization, and abstraction), yielding an updated taxonomy $\taxo_{t+1}$ that remains consistent with the evolving $\repo_{t+1}$.
\end{itemize}
\end{definition}

\subsection{Solution Overview}\label{sec:solution overview}

To address this challenge, we propose \method, a \underline{G}eometric-\underline{I}nferred and \underline{S}elf-refining \underline{T}axonomy Maintenance framework (Fig.~\ref{fig:workflow}) that operates in a continuous three-module cycle:

\begin{itemize}[leftmargin=*, noitemsep,topsep=0pt]

\item \textbf{Cost-effective Evidence Retrieval.}
Guided by the input research topic for initialization ($t=0$) and, for subsequent iterations ($t>0$), by the hypothesized concepts generated by the \textit{Probabilistic Hypothesized Concept Generation} module, it employs a novel \emph{utility-driven, budget-constrained planner} to maximize evidence utility under a user-specified token budget (and corresponding monetary cost): it selects (i) which hypotheses to instantiate as retrieval queries and (ii) which retrieved papers to parse, and then applies an LLM-guided extractor to distill partial taxonomies from their \emph{Related Work} (or equivalent) sections. (Sec.~\ref{sec:paper search}).

\item \textbf{Geometric-inferred Taxonomy Integration.}
This module integrates newly extracted partial taxonomies with the global taxonomy from iteration $t-1$, yielding the updated taxonomy at iteration $t$.
We propose a bidirectional word--box mapping model that embeds concepts in a \textit{geometric box space} and leverages containment relationships to guide robust taxonomy integration.
The model is further designed for reliable and efficient self-supervised training via a dual-objective data selection strategy. (Sec.~\ref{sec:model_and_training})

\item \textbf{Probabilistic Hypothesized Concept Generation.}
This module forecasts likely \emph{structural refinement directions} (e.g., expansion, specialization, and abstraction) in the current taxonomy $\taxo_t$ and generates corresponding \emph{hypothesized concepts} as query candidates for the next iteration.
We instantiate this step with a geometric concept emergence model that leverages the inductive bias of box embeddings: unoccupied regions (geometric ``gaps'') within a parent concept’s embedding box indicate latent space where emerging concepts are likely to reside.
(Sec.~\ref{sec:missing_box})
\end{itemize}

\noindent\textit{Remark:} While we instantiate \method{} for evolving scholarly data, the framework is broadly applicable to other evolving hierarchical knowledge structures whose schema itself drifts over time, including product catalogs, biomedical concept hierarchies, and domain-specific knowledge bases. Adapting to a new domain requires changes only at the evidence-extraction stage: the regex section detector and Related-Work prompt would be replaced by a domain-specific extractor (e.g., a category-attribute parser for e-commerce catalogs, or an ontology-aligned span extractor for biomedical texts). The geometric integration, hypothesized concept generation, and budget-aware retrieval modules are domain-agnostic and require no modification.

\section{Geometric-inferred Taxonomy Integration}
\label{sec:model_and_training}
Merging partial taxonomies with the global taxonomy  presents challenges beyond simple name resolution. 
Entity Resolution methods~\cite{huang2023er1,galhotra2018er2} alone cannot resolve structural conflicts—such as granularity mismatches or cycles—often resulting in structural incoherence. Consequently, we formulate this integration problem within a geometric box-embedding space, utilizing geometric containment as a strict inductive bias for hierarchy construction (Sec.~\ref{sec:box_embedding}).
To bridge the semantic and geometric spaces, we propose a bidirectional word--box mapping model (Sec.~\ref{sec:bidirectional_mappingd_model}).
Crucially, to address the challenges of continual learning and noisy training signals, we employ a self-supervised dual-objective training scheme that selectively filters high-quality signals for robustness while retaining a representative historical set for efficient incremental updates (Sec.~\ref{sec:incremental_training}).

\subsection{Box Representations for Taxonomy Structures}\label{sec:box_embedding}
In the geometric box-embedding space~\cite{Lu2024BoxTM,Xue2024TaxBox}, each concept $v \in \node$ is encoded as an axis-aligned hyper-rectangle (hereafter referred to as a \textit{box}) in a $d$-dimensional Euclidean space, as illustrated in the right panel of Fig.~\ref{fig:bidirectional_mapping_model}. Formally, we adopt a center-offset parameterization:
$$
B_v = B(\mathbf{c}_v, \mathbf{b}_v)
= \left\{\, 
\mathbf{x} \in \mathbb{R}^d \;\middle|\; 
\bigl| x^{(k)} - c_v^{(k)} \bigr| \le b_v^{(k)},\;\forall k \in \{1,\ldots,d\}
\right\}.
$$
Here, $\mathbf{c}_v \in \mathbb{R}^d$ denotes the \emph{center} of the box, and 
$\mathbf{b}_v \in \mathbb{R}^d_{\ge 0}$ is the non-negative \emph{offset} vector 
determining its spatial extent along each dimension. $d$ is the dimensionality of the box-embedding space, and $k$ indexes each coordinate axis from $1$ to $d$.

In the geometric space, hypernym-hyponym relations are encoded via containment constraints. Specifically, for any directed edge $(p, q) \in \edge$, the parent concept $p$ must spatially \textit{contain} its child concept $q$ (i.e., $B_p \supseteq B_q$). 
For example, in Fig.~\ref{fig:workflow}, the box corresponding to \textit{Generation Phase} must enclose those of its child concepts, including \textit{Hard Prompting} and \textit{Soft Prompting}.

\subsection{Bidirectional Word--Box Mapping Model}\label{sec:bidirectional_mappingd_model}
\subsubsection{Model Architecture}
In order to bridge the semantic space and geometric (box embedding) space, we propose a bidirectional word--box mapping model that jointly learns two complementary transformations, as illustrated in Fig.~\ref{fig:bidirectional_mapping_model}:

\begin{itemize}[leftmargin=*,noitemsep,topsep=0pt]
\item \textbf{Word$\to$Box (\wtob):} Given a concept $v$, a frozen pretrained encoder (e.g., BERT~\cite{devlin2019bert}) first converts it into a word embedding $\mathbf{w}$. To capture the non-linear transformation from the semantic space to the geometric space, we employ a multilayer perceptron (MLP) that maps $\mathbf{w}$ to a joint latent representation of dimension $2d$. We split this output to obtain the center $\mathbf{c}$ and an offset vector $\mathbf{b}$:
$
[\mathbf{c}; \mathbf{b}] = \text{MLP}_{\text{enc}}(\mathbf{w})
$
where $[\cdot ; \cdot]$ denotes vector concatenation.

\item \textbf{Box$\to$Word (\btow):}
Similarly, a second MLP implements the reverse transformation by projecting a box embedding $B(\mathbf{c}, \mathbf{b})$ back into a continuous query embedding in the semantic space. We concatenate the center and offset vectors to form a unified geometric feature, which is then mapped to a semantic vector $\mathbf{w}$ via a decoder MLP:
$
\mathbf{w} = \text{MLP}_{\text{dec}}([\mathbf{c}; \mathbf{b}]).
$
\end{itemize}

\begin{figure}[t]
    \centering
    \includegraphics[width=0.6\linewidth]{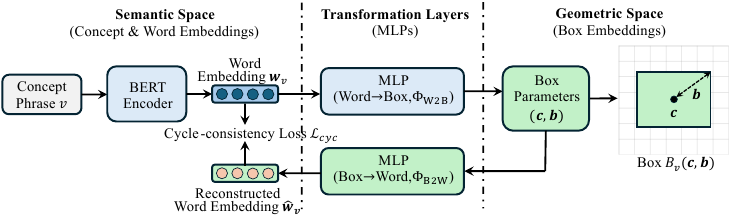}
    \caption{Bidirectional Word--Box Mapping Model.}
    \label{fig:bidirectional_mapping_model}
\end{figure}

\subsubsection{Training Objectives}
The model is trained using self-supervised signals derived from the extracted parent-child relations. We optimize a joint objective consisting of three loss functions:

\begin{itemize}[leftmargin=*,noitemsep,topsep=0pt]
\item \textbf{Geometric Containment Loss $\mathcal{L}_{\text{cont}}$.}
To enforce hierarchy, we require that the parent box $B_p$ spatially encapsulates its child $B_q$. We quantify this by measuring the volume of the child covered by the parent:
$\operatorname{cov}(B_p, B_q)
= \frac{\vol(B_p \cap B_q)}{\vol(B_q)}
\in [0,1],$
where $\vol(B) = \prod_{k=1}^{d} 2b^{(k)}$. 
The loss is minimized when containment is complete:
$
\mathcal{L}_{\text{cont}}(p,q)
= 1 - \operatorname{cov}(B_p,B_q).
$

\item \textbf{Cycle-Consistency Loss $\mathcal{L}_{\text{cycle}}$.}
Since we lack ground-truth geometric coordinates to supervise the inverse \btow mapping directly, we adopt a self-supervised cycle-consistency approach~\cite{zhu2017cycle1}. We penalize the reconstruction error between the original word embedding $\mathbf{w}$ and its round-trip projection through the box space:
$
\mathcal{L}_{\text{cycle}}(v)
= \frac{1}{2}
\Bigl(
1 - \cos\bigl(
\Phi_{\text{B2W}}(\Phi_{\text{W2B}}(\mathbf{w})),
\mathbf{w}
\bigr)
\Bigr).
$
This ensures that the learned geometric representations retain high semantic fidelity to the input space.
\item \textbf{Volume Regularization Loss $\mathcal{L}_{\text{vol}}$.}
To avoid degenerate solutions where box volumes collapse to zero (leading to trivial containment), we enforce a minimum-volume constraint by penalizing boxes whose volume falls below a threshold $V_{\min}$:
$
\mathcal{L}_{\text{vol}}
=
\frac{1}{|\mathcal{V}|}
\sum_{v \in \mathcal{V}}
\max\!\,\Bigl(0, V_{\min} - \vol(B_v)\Bigr).
$
\end{itemize}

\smallskip
\noindent
\textit{Overall Objective.}
The total loss is a weighted sum:
$
\mathcal{L}
=
\mathcal{L}_{\text{cont}}
+
\mathcal{L}_{\text{cycle}}
+
\lambda_{\text{vol}} \mathcal{L}_{\text{vol}},
$
where $\lambda_{\text{vol}}$ balances the regularization strength against the primary structural and semantic objectives.

\subsection{Robust and Efficient Incremental Training}
\label{sec:incremental_training}
To optimize these objectives without manual annotation, we derive self-supervised signals from parent–child relations in the extracted partial taxonomies. 
However, the continuous influx of papers introduces a \textit{learning dilemma}: Updating the model is necessary to capture structural changes, but doing so without historical data can cause  \emph{catastrophic forgetting}~\cite{kirkpatrick2017catastrophic}; in contrast, retraining on the entire dataset is computationally infeasible. 
This process is further complicated by inherent noise in the self-supervision signals, stemming from three sources: (1) \textit{Subjectivity}, arising from divergent author perspectives and terminology; (2) \textit{Temporal Drift}, where the ground-truth structure itself evolves (e.g., the emergence of new concepts); (3) \textit{Extraction Errors}, caused by inevitable inaccuracies in LLM processing. 
Consequently, relying solely on direct supervision from raw extracted pairs is prone to inconsistency and error accumulation.

To address these issues, \method introduces a \emph{dual-objective data selection strategy} designed for efficient and robust incremental training. This strategy selects two types of data:
1) \textbf{New data} ($\mathcal{D}_{\text{new}}$). A \emph{Reliability-Aware Graph Selection} retains only the most reliable relation pairs, thereby reducing noise in the supervision signal.
2) \textbf{Historical data} ($\mathcal{D}_{\text{old}}$). A \emph{novelty-aware coreset selection} mechanism selects a compact yet representative subset of historical training instances, preserving accumulated knowledge without incurring the cost of retraining from scratch.
\subsubsection{Reliability-Aware Graph Selection ($\mathcal{D}_{\text{new}}$).}
To distill high-quality signals from $\mathcal{D}_{\text{new}}$, we construct a weighted candidate graph from extracted partial taxonomies and then select a subset of high-support relations that ensures broad concept participation.

\noindent\textbf{Graph construction from partial taxonomies.}
To mitigate terminological heterogeneity, we employ Entity Resolution (ER)~\cite{christen2012datamatching} to normalize diverse concept names into canonical identifiers. This process consolidates synonyms, such as \textit{Cardinality Estimation} and \textit{Selectivity Estimation}, via a blocking-and-matching pipeline, instantiated using either deep learning models~\cite{tu2023unicorn,li2020deep} or lightweight LLMs~\cite{llama3}.
We then aggregate the normalized parent--child relations into a weighted \textit{candidate concept graph} $\mathcal{G}_{\text{cand}} = (\mathcal{V}_{\text{cand}}, \mathcal{E}_{\text{cand}})$, where $\mathcal{V}_{\text{cand}}$ denotes the set of unique candidate concepts.
For each directed edge $(u,v)\in \mathcal{E}_{\text{cand}}$, we assign a support weight $w_{uv}$ equal to the number of unique source papers that contain this relation. Larger $w_{uv}$ indicates stronger corroboration across independent evidence sources and is used as a proxy for relation reliability.

\noindent\textbf{Coverage-Reliability Selection via Submodular Maximization.}
Given the set of candidate relations $\mathcal{E}_{\text{cand}}$, we aim to select a subset $\mathcal{D}_{\text{rel}} \subseteq \mathcal{E}_{\text{cand}}$ to serve as reliable supervision signals. This selection is constrained by a cardinality budget and must balance two objectives: (i) \textit{Coverage}, ensuring broad semantic representation by maximizing the number of unique participating concepts; and (ii) \textit{Reliability}, prioritizing relations with strong support weights.
We formulate this trade-off as a constrained maximization problem. 
Let $\text{Cov}(\mathcal{D}_{\text{rel}})$ denote the set of unique concepts covered by the selected relations. 
We define the normalized coverage score as $f_{\text{cov}}(\mathcal{D}_{\text{rel}}) = |\text{Cov}(\mathcal{D}_{\text{rel}})| / |\mathcal{V}_{\text{cand}}|$.
Then, we quantify reliability using log-normalized edge weights to mitigate the effect of heavy-tailed distributions. 
Defining the normalized weight $\tilde{w}_{uv} = \log(1+w_{uv})/\log(1+w_{\max})$, the reliability score is given by $f_{\text{rel}}(\mathcal{D}_{\text{rel}}) = \frac{1}{K_{\mathrm{rel}}}\sum_{(u,v)\in\mathcal{D}_{\text{rel}}}\tilde{w}_{uv}$.
The overall optimization problem is:
\begin{equation}\label{eq:cover_rel}
\begin{aligned}
\max_{\mathcal{D}_{\text{rel}}\subseteq \mathcal{E}_{\text{cand}}}\quad 
& F(\mathcal{D}_{\text{rel}})=f_{\text{cov}}(\mathcal{D}_{\text{rel}})+\lambda_{\text{rel}}\,f_{\text{rel}}(\mathcal{D}_{\text{rel}}) \\
\text{s.t.}\quad 
& |\mathcal{D}_{\text{rel}}|\le K_{\mathrm{rel}}.
\end{aligned}
\end{equation}

where $K_{\mathrm{rel}}=\lfloor \eta_{\text{rel}}|\mathcal{D}_{\text{new}}|\rfloor$ is the cardinality budget determined by the retention fraction $\eta_{\text{rel}}$, and $\lambda_{\text{rel}}$ is a hyperparameter balancing the two terms.

\noindent\textbf{Greedy Optimization and Guarantee.} 
The objective function $F$ is non-negative, monotone, and submodular, as it combines a coverage term exhibiting diminishing marginal returns with a strictly modular (linear) reliability term. (We provide a formal proof in Appendix~\ref{sec:proof_submodularity_f1}.) Leveraging these properties, we employ a standard greedy algorithm that iteratively selects the edge $e^*$ maximizing the marginal gain $\Delta F(e \mid \mathcal{D}_{\text{rel}})$, which is computationally efficient and theoretically grounded, providing a $(1-1/e)$-approximation guarantee to the optimal solution~\cite{nemhauser1978analysis}.

\noindent\textit{Remark.} Reliability-aware filtering grounds GIST's robustness in \emph{cross-paper consensus} rather than any single extraction signal, operating along two complementary axes. \emph{Synchronically}, key concepts and relations within a domain are typically attested by multiple papers, so missing detections (e.g., papers without an explicit Related Work section) and noisy extractions are absorbed by corroborating evidence: uncorroborated relations are downweighted through the support weights $w_{uv}$, while well-supported structures are preserved. \emph{Diachronically}, the same principle yields a \emph{self-correcting mechanism} for niche subfields, where low-support branches may reflect either idiosyncratic framing by a single author or the genuine emergence of a novel domain. The continuous update mechanism resolves this ambiguity \emph{a posteriori}: branches the community adopts accumulate cross-paper support and their $w_{uv}$ grows over time, ensuring survival under selection, while unreinforced branches remain low-weight and are gradually displaced. The filter therefore neither commits prematurely to any single author's framing nor over-penalizes emerging structure, adapting instead to the empirical distribution of cross-paper agreement as the repository grows.

\subsubsection{Incremental Model Update via Novelty-Aware Coreset Selection ($\mathcal{D}_{old}$).}
\label{sec:coreset_update}
To enable continual refinement of the dual mapping model without full retraining, we draw inspiration from prior research on coreset selection~\cite{deng2025two,wang2022coresets,chai2023goodcore}. At each iteration, the full historical dataset $\mathcal{D}_{\text{old}}$ is approximated by a compact, weighted \emph{coreset} $\mathcal{D}_\text{core}$, defined as a subset of size $|\mathcal{D}_\text{core}| = \eta_{\text{core}}|\mathcal{D}_\text{old}|$, where $\eta_{\text{core}} \in (0,1]$ is a hyperparameter controlling the compression ratio. This coreset is then utilized for subsequent training.

\noindent\textbf{Primary Goal: Loss Preservation.} 
The primary goal of the coreset adheres to the principles of \textit{Empirical Risk (Loss) Minimization (ERM)}~\cite{mohri2018foundations}, which aims to select $\mathcal{D}_{\text{core}}$ with weights $\{w'_e\}_{e \in \mathcal{D}_{\text{core}}}$ such that its weighted empirical risk $\tilde{\mathcal{L}}_{\text{old}}(\theta)=\sum_{e \in \mathcal{D}_{\text{core}}} w'_e \ell(e; \theta) $  closely approximates the full historical risk $\mathcal{L}_{\text{old}}(\theta)$:
\begin{equation}
\min_{\mathcal{D}_{\text{core}} \subset \mathcal{D}_{\text{old}}, \{w'_e\}} \left| \mathcal{L}_{\text{old}}(\theta) - \sum_{e \in \mathcal{D}_{\text{core}}} w'_e \, \ell(e; \theta) \right|
\end{equation}
where $w'_e \ge 0$ (with $\sum w'_e = 1$) denotes the importance weight assigned to each selected instance to mirror the full-data objective.
The key issue is thus to choose an
optimal sampling distribution over $\mathcal{D}_{\text{old}}$ from which to
construct $\mathcal{D}_{\text{core}}$ and assign the weights
$\{w'_e\}$.

\noindent\textbf{Complementary Goal: Novelty for Dynamic Data.}
Traditional coreset selection assumes a static dataset, ignoring potential redundancy between the historical coreset and incoming data $\mathcal{D}_{\text{new}}$. We address this by enforcing a \textit{novelty} constraint: \emph{retained historical instances must be complementary to the new batch}.
To quantify this, we define a novelty score $n(e)$ for each historical edge $e=(p, q) \in \mathcal{D}_{\text{old}}$. First, we represent each edge as the concatenation of its parent and child word embeddings ($\mathbf{w}_p, \mathbf{w}_q$):
$
\mathbf{r}(e) = [\mathbf{w}_p; \mathbf{w}_q].
$
We then measure the maximum similarity of $e$ to the incoming batch $\mathcal{D}_{\text{new}}$:
$
s(e) = \max_{e' \in \mathcal{D}_{\text{new}}} \cos\left( \mathbf{r}(e),\;\mathbf{r}(e') \right).
$
The final novelty score is $n(e) = 1 - s(e)$, which prioritizes historical instances that are semantically distinct from the new data.

\noindent\textbf{Novelty-Aware Coreset Selection.} Based on the above idea, we formulate coreset selection as an optimization problem to find a sampling distribution that balances fidelity to the historical risk ($\Pi_{\text{ERM}}$) with novelty relative to the new data. 
\begin{equation}
\label{eq:nov-kl}
\max_{\Pi_{\text{core}} \in \Delta(\mathcal{D}_{\text{old}})}
-\text{KL}\!\left(\Pi_{\text{core}} \,\|\, \Pi_{\text{ERM}}\right)
\;+\;
\lambda_{\mathrm{nov}}\,\mathbb{E}_{e \sim \Pi_{\text{core}}}\!\left[n(e)\right],
\end{equation}
where $\Delta(\mathcal{D}_{\text{old}}) = \{\Pi_{\text{core}} : \mathcal{D}_{\text{old}} \rightarrow [0,1] \mid \sum_{e} \Pi_{\text{core}}(e) = 1\}$ denotes the set of all valid probability distributions over the $\mathcal{D}_{\mathrm{old}}$. Here, $\Pi_{\text{ERM}}$ represents the loss distribution at the previous timestamp, derived directly from the preceding training stage. $\text{KL}(\cdot \,\|\, \cdot)$ denotes the Kullback--Leibler divergence, and $\lambda_{\mathrm{nov}}$ is a hyperparameter controlling the trade-off between loss preservation and novelty.

To derive the optimal sampling distribution $\Pi_{\text{core}}$ and the corresponding importance weights $w'_e$, we establish Propositions~\ref{prop:coreset_sampling} and~\ref{prop:coreset_weights}. 
Detailed proofs are provided in Appendix~\ref{sec:proof_coreset_sampling} and~\ref{sec:proof_coreset_weights}.

\begin{proposition}[Optimal Sampling Distribution]
\label{prop:coreset_sampling}
The unique sampling distribution $\Pi_{\text{core}}^{\ast}$ that maximizes the novelty-aware objective in Eq.~\ref{eq:nov-kl} is given by
$
\Pi_{\text{core}}^{\ast}(e)
=
\frac{\Pi_{\text{ERM}}(e)\exp\{\lambda_{\mathrm{nov}} n(e)\}}
{\sum_{e'} \Pi_{\text{ERM}}(e')\exp\{\lambda_{\mathrm{nov}} n(e')\}}.
$
\end{proposition}

\begin{proposition}[Optimal Coreset Weighting]
\label{prop:coreset_weights}
To ensure that the coreset risk $\tilde{\mathcal{L}}_{\mathrm{old}}(\theta)$ is a consistent estimator of the full historical risk $\mathcal{L}_{\mathrm{old}}(\theta)$, each instance $e \in \mathcal{D}_{\mathrm{core}}$ is assigned a normalized inverse-probability weight
$
w'_e \;=\; \frac{1 / \Pi_{\mathrm{core}}^{*}(e)}{\sum_{e_j \in \mathcal{D}_{\mathrm{core}}} 1 / \Pi_{\mathrm{core}}^{*}(e_j)} \, .
$
\end{proposition}

\subsection{Taxonomy Integration}
\label{sec:taxonomy_integration}

As presented in Alg.~\ref{alg:taxonomy_integration}, we construct the unified global taxonomy $\mathcal{T}_t$ by merging concepts from the previous state $\mathcal{T}_{t-1}$ with novel concepts from the partial taxonomies at timestamp $t$. 
Specifically, we form a unified concept set $\mathcal{V}_t=\mathcal{V}(\mathcal{T}_{t-1}) \cup \bigcup_{i=1}^{n}\mathcal{V}(\mathcal{T}_{par,i})$ and project each concept $v\in\mathcal{V}_t$ into the box-embedding space via $B_v=\Phi_{\text{W2B}}(\mathbf{w}_v)$ (Lines~1--2). 
We then infer candidate hypernym--hyponym relations by evaluating geometric containment for each ordered pair of concepts (Lines~4--9). 
Concretely, we compute the containment score $s_{u\rightarrow v}=\operatorname{cov}(B_u,B_v)$ and introduce a candidate edge $u\rightarrow v$ only if $s_{u\rightarrow v}\ge\tau$ (Lines~5--7), where $\tau$ is a predefined threshold (set to $0.9$ by default).
To ensure that the candidate relation set is acyclic under approximate containment scores, we retain only edges directed from larger to smaller boxes, i.e., $\mathrm{Vol}(B_u)>\mathrm{Vol}(B_v)$ (Line~6). Finally, since containment is transitive (e.g., $B_c \subset B_b \subset B_a$), we apply transitive reduction to remove redundant indirect relations and obtain a minimal, consistent hierarchy $\mathcal{T}_t$ (Line~10).

\begin{algorithm}[t]
\caption{Taxonomy Integration}
\label{alg:taxonomy_integration}
\footnotesize
\begin{algorithmic}[1]
\REQUIRE Previous global taxonomy $\mathcal{T}_{t-1}$; partial taxonomies $\{\mathcal{T}_{par,1},\dots,\mathcal{T}_{par,n}\}$; $\Phi_{\text{W2B}}$; containment threshold $\tau$.
\ENSURE Updated global taxonomy $\mathcal{T}_t$.
\STATE Construct the unified concept set $\mathcal{V}_t \leftarrow \mathcal{V}(\mathcal{T}_{t-1}) \cup \bigcup_{i=1}^{n}\mathcal{V}(\mathcal{T}_{par,i})$.
\STATE For each $v \in \mathcal{V}_t$, compute its box embedding $B_v \leftarrow \Phi_{\text{W2B}}(\mathbf{w}_v)$.
\STATE Initialize the candidate edge set $\mathcal{E}_{\text{cand}} \leftarrow \emptyset$.
\FOR{each ordered pair $(u, v) \in \mathcal{V}_t \times \mathcal{V}_t$, $u \neq v$}
    \STATE Compute the geometric containment score $s_{u \rightarrow v} \leftarrow \operatorname{cov}(B_u, B_v)$.
    \IF{$s_{u \rightarrow v} \ge \tau$ \textbf{and} $\mathrm{Vol}(B_u) > \mathrm{Vol}(B_v)$}
        \STATE Add $(u, v)$ to $\mathcal{E}_{\text{cand}}$.
    \ENDIF
\ENDFOR
\STATE Apply transitive reduction to obtain a minimal edge set $\mathcal{E}_{\text{final}} \leftarrow \operatorname{TransitiveReduction}(\mathcal{V}_t, \mathcal{E}_{\text{cand}})$.
\RETURN $\mathcal{T}_t \leftarrow (\mathcal{V}_t, \mathcal{E}_{\text{final}})$.
\end{algorithmic}
\end{algorithm}

\section{Probabilistic Hypothesized Concept Generation}\label{sec:missing_box}
Based on the global taxonomy $\mathcal{T}_t$, we propose a Probabilistic Hypothesized Concept Generation module that produces a set of hypothesized concepts for iteration $t+1$. 
These hypotheses serve as query candidates for the subsequent evidence retrieval stage: by targeting likely emerging topics, they provide more informative search signals than reusing existing taxonomy concepts, enabling \emph{cost-effective retrieval}.
To this end, we introduce a \textit{Geometric Concept Emergence Model} grounded in the geometric box embedding space, where containment constraints act as an intrinsic inductive bias for concept emergence under different structural refinements (Sec.~\ref{sec:concept_emergence}). 
Finally, we enable efficient online updates via a self-supervised, lightweight Method-of-Moments estimator that infers parameters directly from the taxonomy structure (Sec.~\ref{sec:parameter estimation}).

\subsection{Geometric Concept Emergence Model}\label{sec:concept_emergence}
Given the taxonomy at iteration $t$, our goal is to estimate $P(v \mid \mathcal{T}_t)$---the probability that a new concept $v \in \Delta \mathcal{V}_{t+1}=\mathcal{V}_{t+1}\setminus \mathcal{V}_t$ emerges at $t{+}1$.
We model emergence as two coupled decisions:
\begin{itemize}[leftmargin=*,noitemsep,topsep=0pt]
\item \textit{Structural Refinement Context Selection.}
We first select a structural refinement context $C \in \mathbb{C}(\mathcal{T}_t)$, which specifies (i) an anchor (parent) node $p$ together with (when applicable) the subset of its children involved, and (ii) the refinement operator to apply. 
For example, at $t_3$ in Example~\ref{ex:motivating_example}, the structural refinement context corresponds to applying the refinement operator \emph{abstraction} to the anchor parent node \emph{Generation Phase} and its existing children \emph{Flatten} and \emph{Natural Language}.

\item \textit{Context-aware Generation.}
Conditioned on the selected context $C$, we then generate the specific emerging concept $v$ that is compatible with the geometric constraints implied by $C$.
\end{itemize}

Thus, we introduce a latent random variable $\mathcal{C}$ ranging over the refinement-context set $\mathbb{C}(\mathcal{T}_t)$, and factorize the concept \emph{emergence} probability by marginalizing over all possible contexts:
$
P(v \mid \mathcal{T}_t)
=
\sum_{C \in \mathbb{C}(\mathcal{T}_t)}
P\!\left(v \mid \mathcal{C}=C\right)\,
P\!\left(\mathcal{C}=C \mid \mathcal{T}_t\right).
$
Moreover, we adopt the modeling assumption that, for any candidate concept $v$, there is at most one \emph{compatible} context $C$ such that
$P(v \mid \mathcal{C}=C)>0$.
For instance, \emph{Local Retrieval} can arise under the \emph{Retrieval Phase} but is incompatible with the \emph{Indexing Phase}.
While a concept could in principle fit multiple contexts, this assumption holds for the vast majority of cases in practice and keeps the factorization tractable.
Hence, the summation collapses to the single valid term:
\begin{equation}\label{eq:context_collapse}
P(v \mid \mathcal{T}_t)
=
P\!\left(v \mid \mathcal{C}=C\right)\,
P\!\left(\mathcal{C}=C \mid \mathcal{T}_t\right),
\end{equation}
where $C$ denotes the (unique, if it exists) context compatible with $v$.

We instantiate this factorization in a geometric box-embedding space by representing each concept $v$ with a box embedding $B$ and encoding structural relations via containment constraints. This yields two modeling objectives:
(1) \emph{context-aware generation} $P(B \mid \mathcal{C}=C)$, which characterizes how an emerging concept can be geometrically realized under the constraints induced by $C$; and
(2) \emph{context selection prior} $P(\mathcal{C}=C \mid \mathcal{T}_t)$, which prioritizes the most likely refinement contexts under the current taxonomy.

\subsubsection{Context-aware Generation.}\label{sec:context_conditioned_generation}
The key advantage of modeling emergence in the geometric box space is that geometric relations impose strong inductive biases, enabling us to infer plausible emerging nodes under different structural refinement contexts. As illustrated in Fig.~\ref{fig:gaps}: (1) \textit{Specialization} inserts a new child under an existing node, requiring the child box to be contained within the parent box; (2) \textit{Abstraction} inserts an intermediate node to group a set of siblings, requiring the new box to enclose the grouped children while remaining contained within the original parent; and (3) \textit{Expansion} inserts a new child under a parent with existing children, requiring the new box to be contained within the parent while remaining geometrically separated from existing sibling boxes.

Leveraging these geometric inductive biases, we cast context-aware generation as sampling a valid box $B$ from the feasible region induced by the refinement context $C$, which reduces $P(B\mid \mathcal{C}=C)$ to learning distributions over unit intervals $[0,1]$.

\noindent\textbf{Parent-normalized Reparameterization}. We apply a scale-invariant, parent-normalized reparameterization to make the distribution tractable. 
We re-parameterize the problem into a parent-normalized coordinate system anchored at the context anchor (parent) node $p$. For any box $B_q=(\boldsymbol{c}_q,\boldsymbol{b}_q)$ such that $B_q \subseteq B_p$, we define:
\begin{equation}
\label{eq:normalization}
\boldsymbol{\rho}_q = \boldsymbol{b}_q \oslash \boldsymbol{b}_p \in (0,1]^d,\qquad
\boldsymbol{\delta}_q = (\boldsymbol{c}_q - \boldsymbol{c}_p) \oslash \boldsymbol{b}_p \in [-\boldsymbol{1},\boldsymbol{1}]^d,
\end{equation}
where $\oslash$ denotes element-wise division. Here, $\boldsymbol{\rho}_q$ and $\boldsymbol{\delta}_q$ denote the \emph{parent-normalized offset} (relative box size) and the \emph{parent-normalized center} (relative placement) of $B_q$, respectively.

\noindent\textbf{Geometric Coupling and Sequential Factorization}. Directly modeling the joint distribution $P(\boldsymbol{\rho}_v, \boldsymbol{\delta}_v \mid \mathcal{C}=C)$ is still problematic, since treating size and center independently ignores their strict geometric coupling, often yielding invalid boxes that violate containment boundaries.
To enforce these constraints, we decompose the generative objective via the chain rule:
\begin{equation}\label{eq:generative process decomposition}
P(B \mid \mathcal{C}=C)
=
P(\boldsymbol{\delta}_v \mid \boldsymbol{\rho}_v, \mathcal{C}=C)
\cdot
P(\boldsymbol{\rho}_v \mid \mathcal{C}=C).
\end{equation}
This factorization corresponds to a sequential process: we first sample the relative \emph{offset} ($\boldsymbol{\rho}_v$), and then sample the relative \emph{center} ($\boldsymbol{\delta}_v$) conditioned on that size to ensure geometric validity.

\noindent\textbf{Feasible-Interval Modeling}. Given a context $C$, the feasible set $\mathcal{F}(C)$ admits a per-dimension interval characterization: for each dimension $k$, $(\boldsymbol{\rho}_v,\boldsymbol{\delta}_v)\in \mathcal{F}(C)$ implies
$
\rho_v^{k} \in \big[\underline{\rho}^{k},\,\overline{\rho}^{k}\big],
\;
\delta_v^{k} \in \big[\underline{\delta}^{k},\,\overline{\delta}^{k}\big],
$
where the bounds are induced by $C$. Notably, the feasible interval of $\delta_v^{k}$ depends on the chosen $\rho_v^{k}$: \textit{larger boxes admit a smaller range of valid centers under containment/enclosure constraints.}

This interval motivates an auxiliary-variable reparameterization. For each dimension $k$, we introduce unit random variables $u_{\rho}^{k},u_{\delta}^{k}\in[0,1]$ and map them affinely to the context-induced feasible intervals. Specifically, we parameterize $\rho^{k}_v$ as
$
\label{eq:rho_interval_map}
\rho_v^{k}
=
\underline{\rho}^{k} + u_{\rho}^{k}\big(\overline{\rho}^{k}-\underline{\rho}^{k}\big),
\;\;
u_{\rho}^{k}\in[0,1],
$
and, conditioned on the sampled $\rho_v^{k}$, we parameterize the placement coordinate as
$
\label{eq:delta_interval_map}
\delta_v^{k}
=
\underline{\delta}^{k} + u_{\delta}^{k}\big(\overline{\delta}^{k}-\underline{\delta}^{k}\big),
\qquad
u_{\delta}^{k}\in[0,1].
$
We model these unit variables with the Beta distribution $\mathrm{Beta}(\alpha, \beta)$ for two reasons: its native support on $[0,1]$ guarantees geometric validity without clipping, and its shape flexibility captures diverse density profiles, from uniform uncertainty to peaked confidence, within a single probabilistic framework.
We apply this formulation to all context types $\textsc{ctx} \in \{\textsc{spec}, \textsc{abs}, \textsc{exp}\}$, with distinct Beta parameters for $u_{\rho,\textsc{ctx}}^{k}$ and $u_{\delta,\textsc{ctx}}^{k}$. The problem thus reduces to specifying the boundary constraints $[\underline{\rho}, \overline{\rho}]$ and $[\underline{\delta}, \overline{\delta}]$ per context:

\begin{figure}[t]
    \centering
    \includegraphics[width=0.6\linewidth]{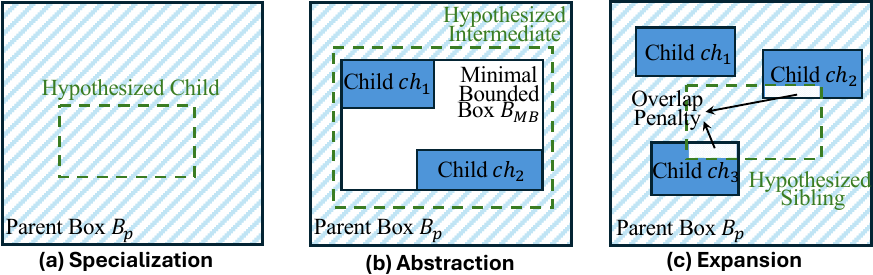}
    \caption{Context-aware Generation.}
    \label{fig:gaps}
\end{figure}

\begin{itemize}[leftmargin=*,noitemsep,topsep=0pt]
\item \textbf{Specialization.}
The new box $B$ must be contained by the parent box $B_p$, imposing $B \subseteq B_p$. Accordingly, the relative size is bounded by the parent’s unit extent, $[\underline{\rho}^{k}, \overline{\rho}^{k}] = [0, 1]$, and the center is constrained to remain within the residual room, yielding $[\underline{\delta}^{k}, \overline{\delta}^{k}] = [-(1-\rho_v^{k}), 1-\rho_v^{k}]$.

\item \textbf{Abstraction.}
Let $B_{\mathrm{MB}}=(\boldsymbol{\delta}_{\mathrm{MB}},\boldsymbol{\rho}_{\mathrm{MB}})$ denote the minimal bounding (MB) box of the grouped children.
The new box must satisfy
$
B_{\mathrm{MB}} \subseteq B \subseteq B_p,
$ 
which yields: (1) \textit{Size Bound:}
$
[\underline{\rho}^{k},\,\overline{\rho}^{k}] = [\rho_{\mathrm{MB}}^{k},\,1];
$
and (2) \textit{Center Bound:} conditioned on $\rho_v^{k}$,
$
[\underline{\delta}^{k},\,\overline{\delta}^{k}]
=
\Big[
\max\!\big\{-1+\rho_v^{k},\;\delta_{\mathrm{MB}}^{k}+\rho_{\mathrm{MB}}^{k}-\rho_v^{k}\big\},\;
\min\!\big\{1-\rho_v^{k},\;\delta_{\mathrm{MB}}^{k}-\rho_{\mathrm{MB}}^{k}+\rho_v^{k}\big\}
\Big].
$

\item \textbf{Expansion.}
This context imposes two constraints:
\textit{(1) Parent-child containment.} The new node must lie inside the parent box. In this case, the feasible intervals coincide with the specialization case. Accordingly, we define a \textit{base proposal distribution} $P_{\text{base}}(B \mid \mathcal{C}=C)$ that adopts the same parametric form as specialization but uses distinct Beta parameters;
\textit{(2) Sibling overlap penalty.} Unlike specialization, $v$ should avoid regions already explained by existing siblings, since taxonomic principles encourage siblings to be semantically distinct rather than redundant~\cite{taxonomy_method}.
Let $\mathcal{B}_{\mathrm{Ch}(p)} = \{B_s\}_{s\in \mathrm{Ch}(p)}$ denote the existing child boxes under $p$. We quantify overlap using the IoU between $B$ and the union of siblings, denoted as $\mathcal{O}(B, \mathcal{B}_{\mathrm{Ch}(p)})$, and define
$
\label{eq:sibling_reweighted}
P(B \mid \mathcal{C}=C)
\propto
P_{\text{base}}(B \mid \mathcal{C}=C)\,
\exp\!\left(-\mathcal{O}(B, \mathcal{B}_{\mathrm{Ch}(p)})/\tau_{\mathrm{sib}}\right),
$
where $\tau_{\mathrm{sib}} > 0$ controls the penalty strength.
\end{itemize}

\subsubsection{Context Selection Prior.}
The Context Selection Prior $P(\mathcal{C}=C \mid \mathcal{T}_t)$ quantifies the likelihood that a candidate context $C$ will be refined at iteration $t+1$. 
We premise this on the geometric intuition that insertion likelihood is proportional to the \emph{void volume} available within that context. We define $\vol_{\mathrm{void}}(C)$ for each structural pattern: 
1) \textit{Specialization void.} For a leaf node $p$ becoming a parent:
$\vol_{\mathrm{void}}(C) = \mathrm{Vol}(B_p).$
2) \textit{Abstraction void.} For grouping children under $p$:
$\vol_{\mathrm{void}}(C) = \mathrm{Vol}(B_p) - \mathrm{Vol}(B_{\mathrm{MB}}).$
3) \textit{Expansion void.} For adding a sibling under $p$:
$\vol_{\mathrm{void}}(C) = \mathrm{Vol}(B_p) - \mathrm{Vol}\!\left(\textstyle\bigcup_{s \in \mathrm{Ch}(p)} B_s\right).$

Finally, we normalize over all candidate contexts $C_i \in \mathbb{C}(\mathcal{T}_t)$ to obtain a valid distribution:
$
\label{eq:context_realization_void}
P(\mathcal{C}=C_i \mid \mathcal{T}_t)
=
\frac{\vol_{\mathrm{void}}(C_i)}{\sum_{C_j \in \mathbb{C}(\mathcal{T}_t)} \vol_{\mathrm{void}}(C_j)}.
$

\subsection{Self-supervised Parameter Estimation}
\label{sec:parameter estimation}
To instantiate the Beta parameters in $P(v \mid \mathcal{T}_t)$, we propose a self-supervised estimation strategy on the current taxonomy $\mathcal{T}_t$. 
Specifically, we generate \emph{pseudo-observations} by masking existing nodes to emulate structural refinement events (Sec.~\ref{sec:pseudo-observation}), and fit the corresponding Beta distributions via a weighted Method-of-Moments (MoM) estimator~\cite{West1979welford}, enabling efficient and incremental updates.

\begin{algorithm}[t]
\caption{Unit-Sample Recovery from Pseudo-Observations}
\label{alg:unit_recovery}
\footnotesize
\begin{algorithmic}[1]
\REQUIRE Pseudo-observation $z_i$; parent-normalized geometry $(\boldsymbol{\rho}_i,\boldsymbol{\delta}_i)$;
\ENSURE Unit samples $(\boldsymbol{u}_{\rho,i},\boldsymbol{u}_{\delta,i}) \in [0,1]^d$ and weight $\omega_i$.

\IF{$\textsc{ctx}_i = \textsc{abs}$}
    \STATE Compute the minimal bounding box: $(\boldsymbol{\delta}_{\mathrm{MB}},\boldsymbol{\rho}_{\mathrm{MB}})\leftarrow \mathrm{MB}(\cdot)$.
\ENDIF

\FOR{each dimension $k=1,\dots,d$}
    \STATE \textbf{Bounds:} derive the per-dimension feasible intervals $[\underline{\rho}^{k},\overline{\rho}^{k}]$ and $[\underline{\delta}^{k},\overline{\delta}^{k}]$ implied by $\textsc{ctx}_i$ and the local context geometry.
    \STATE \textbf{Invert to unit samples:}
    $
    u_{\rho,i}^{k} \leftarrow \frac{\rho_i^{k}-\underline{\rho}^{k}}{\overline{\rho}^{k}-\underline{\rho}^{k}},
    u_{\delta,i}^{k} \leftarrow \frac{\delta_i^{k}-\underline{\delta}^{k}}{\overline{\delta}^{k}-\underline{\delta}^{k}} .
    $
\ENDFOR

\STATE Set $\omega_i \leftarrow 1$.
\IF{$\textsc{ctx}_i = \textsc{exp}$}
    \STATE $\omega_i \leftarrow \exp(\mathcal{O}(B_i)/\tau_{\mathrm{sib}})$.
\ENDIF
\RETURN $(\boldsymbol{u}_{\rho,i},\boldsymbol{u}_{\delta,i}, \omega_i)$.
\end{algorithmic}
\end{algorithm}
\subsubsection{Pseudo-Observation Generation}
\label{sec:pseudo-observation}
Since novel nodes in $\mathcal{T}_{t+1}$ are unobserved at time $t$, we estimate these parameters via self-supervision on the current taxonomy $\mathcal{T}_t$. We construct pseudo-observations $\{z_i\}$ by selecting a \emph{context} from $\mathcal{T}_t$ and temporarily \emph{removing} a node that naturally fits the corresponding refinement pattern; the remaining structure induces a feasible region for that node. We consider three masking strategies:
1) \textbf{Specialization.} For each parent $p$ whose only child is $v$ (so that removing $v$ leaves $p$ a leaf, mirroring the specialization pattern), we take the child context at $p$ and remove $v$.
2) \textbf{Abstraction.} For each chain $v_1 \rightarrow v_2 \rightarrow \mathrm{Ch}(v_2)$, we take the intermediate context between $v_1$ and $\mathrm{Ch}(v_2)$ and remove the intermediate node $v_2$.
3) \textbf{Expansion.} For each parent $p$ with at least two children, we take the sibling context under $p$ and perform leave-one-out by removing one $v\in \mathrm{Ch}(p)$; the siblings that remain in the context distinguish this case from specialization. This process yields a collection of pseudo-observations $\{z_i=(\mathcal{C}_i, v_i)\}$, where $\mathcal{C}_i$ denotes the induced context and $v_i$ denotes the masked target node.

\subsubsection{Parameter Estimation via Weighted Method of Moments}
We estimate parameters of the unit variables via a \emph{weighted} Method-of-Moments (MoM), which matches the weighted empirical mean and variance of unit samples to the corresponding theoretical moments. The estimation proceeds in two stages (Alg.~\ref{alg:unit_recovery} and Alg.~\ref{alg:online_mom_beta}).

\noindent\textbf{Stage I (Alg.~\ref{alg:unit_recovery}): unit-sample recovery and debiasing.}
For each pseudo-observation $z_i$, this stage first computes the context-induced size and center bounds per dimension including the auxiliary information required for the $\ctxabs$ context (Alg.~\ref{alg:unit_recovery}, Line~1--5), and then inverts the affine mapping to obtain unit samples in $[0,1]$ (Line~6).
For the \textsc{exp} context, since the overlap penalty biases observations, the inverse-probability weighting $\omega_i=\exp(\mathcal{O}(B_i)/\tau_{\mathrm{sib}})$ is applied to recover moments of the underlying base proposal (Lines~9--12).

\noindent\textbf{Stage II (Alg.~\ref{alg:online_mom_beta}): online weighted moment updates.}
We stream the weighted unit samples and maintain running statistics---$\Omega$ (weight sum), $\Omega_{\mathrm{sq}}$ (sum of squared weights), $\mu_\omega$ (weighted mean), and $M_{2,\omega}$ (weighted second central moment)---via a weighted Welford update~\cite{West1979welford} (Alg.~\ref{alg:online_mom_beta}, Lines~1--7), and then compute $(\hat{\alpha},\hat{\beta})$ in closed form (Lines~9--10). This yields incremental, up-to-date priors for each context and dimension without revisiting past data.

\begin{algorithm}[t]
\caption{Online Weighted MoM for $\mathrm{Beta}(\alpha,\beta)$}
\label{alg:online_mom_beta}
\footnotesize
\begin{algorithmic}[1]
\REQUIRE Stream of weighted unit samples $\{(u_j,\omega_j)\}$ with $u_j\in[0,1]$.
\ENSURE MoM estimates $(\hat{\alpha},\hat{\beta})$.
\STATE Maintain running statistics: $\Omega=\sum_j \omega_j$, $\Omega_{\mathrm{sq}}=\sum_j \omega_j^2$, $\mu_\omega$, and $M_{2,\omega}$.
\STATE Initialize $\Omega, \Omega_{\mathrm{sq}}, \mu_\omega, M_{2,\omega} \leftarrow 0$.
\FOR{each $(u_j,\omega_j)$}
    \STATE $\Omega^{\mathrm{new}} \leftarrow \Omega + \omega_j$;\;\; $\Delta \leftarrow u_j - \mu_\omega$.
    \STATE $\mu_\omega^{\mathrm{new}} \leftarrow \mu_\omega + \frac{\omega_j}{\Omega^{\mathrm{new}}}\Delta$.
    \STATE $M_{2,\omega} \leftarrow M_{2,\omega} + \omega_j\,\Delta\,(u_j-\mu_\omega^{\mathrm{new}})$.
    \STATE $\Omega_{\mathrm{sq}} \leftarrow \Omega_{\mathrm{sq}} + \omega_j^2$;\;\; $\Omega \leftarrow \Omega^{\mathrm{new}}$;\;\; $\mu_\omega \leftarrow \mu_\omega^{\mathrm{new}}$.
\ENDFOR
\STATE $\hat{v}_\omega \leftarrow \frac{\Omega\,M_{2,\omega}}{\Omega^2-\Omega_{\mathrm{sq}}}$;\;\; $\kappa \leftarrow \frac{\mu_\omega(1-\mu_\omega)}{\hat{v}_\omega}-1$.
\STATE $\hat{\alpha} \leftarrow \mu_\omega\kappa$;\;\; $\hat{\beta} \leftarrow (1-\mu_\omega)\kappa$.
\RETURN $(\hat{\alpha},\hat{\beta})$.
\end{algorithmic}
\end{algorithm}

\section{Cost-effective Evidence Retrieval}\label{sec:paper search}
When a new iteration $t$ begins and new research papers arrive, we leverage the hypothesized concepts derived in the previous iteration as \emph{query candidates} for evidence retrieval. 
The key challenge is to derive an \emph{optimal query plan}: decide which candidate structural refinement contexts $\mathcal{C}_h$ to explore and how much budget to allocate to each under a strict token budget.
A naive Monte Carlo strategy will repeatedly sample candidate concepts from each hypothesis distribution $P(B \mid \mathcal{T}_{t-1})$, which leads to two key limitations.
First, since hypotheses are inferred solely from $\mathcal{T}_{t-1}$, some may not be supported by the repository, wasting budget on irrelevant retrievals.
Moreover, the strategy is both costly and high-variance: it requires many samples for stable behavior, and retrieval quality can be dominated by sampling noise.
To address these issues, we propose a novel budget-aware utility-driven planner that anchors each candidate context to the repository and estimates its expected relevance utility and expected parsing cost, yielding an optimal budget-feasible query plan (Sec.~\ref{sec:planning}).
Then, we materialize this plan to retrieve a final set of papers and extract partial taxonomies from them. (Sec.~\ref{sec:execution}).

\subsection{Budget-aware Utility-driven Planner}\label{sec:planning}
\subsubsection{Query Planning as an Optimization Problem.}
At iteration $t$, we consider the candidate structural refinement contexts extracted from the taxonomy at the previous timestamp
$\mathbb{C}(\mathcal{T}_{t-1})=\{\mathcal{C}_h\}_{h=1}^{H_{t-1}}$, together with the newly arrived papers $\Delta \repo_t$.
The planner aims to derive a query-allocation policy $\Pi_t$ over these contexts that jointly accounts for the following objectives and constraints:

\noindent\textbf{Primary Objective: Utility Anchoring.}
The planner \emph{anchors} each candidate context $\mathcal{C}_h$ by estimating its empirical evidence utility $U_h$, which captures how strongly the available papers can support the corresponding hypothesis. Under an allocation policy $\Pi_t$, the expected relevance gain is
$
\sum_{h=1}^{H_{t-1}}\Pi_{t,h}\,U_h,
$
which favors allocating more budget to contexts whose hypotheses are better supported by the available papers. We detail the utility function $U$ later in Sec.~\ref{sec:gaussian_proxy}.

\noindent\textbf{Secondary Objective: Context Selection Prior.}
Beyond utility maximization, we impose the \emph{context realization prior} that biases exploration toward contexts that are more likely to admit valid emerging concepts under the current taxonomy $\mathcal{T}_{t-1}$.
Let $A_h\in\{0,1\}$ indicate whether context $\mathcal{C}_h$ is selected under $\mathcal{T}_{t-1}$.
We instantiate this preference via the prior allocation
$
\Pi^{\mathrm{prior}}_{t-1,h}
=
P\!\left(A_h{=}1 \mid \mathcal{T}_{t-1}\right),
$
and regularize the planning policy $\Pi_t$ toward $\Pi_{t-1}^{\mathrm{prior}}$ through
$\mathrm{KL}\!\big(\Pi_t\,\|\,\Pi_{t-1}^{\mathrm{prior}}\big)$.

\noindent\textbf{Budget Constraint: Token-Feasible Planning.}
The per-iteration token budget $K_t$ limits the \emph{total} parsing cost incurred when reading related-work sections in iteration $t$.
Since different contexts may retrieve papers of different lengths, we summarize the context-induced parsing burden by an \emph{expected per-paper cost} $\bar{\ell}_h$. Under allocation $\Pi_t$, the expected cost per parsed paper is the mixture
$
\sum_{h=1}^{H_{t-1}}\Pi_{t,h}\bar{\ell}_h.
$
We first estimate a token-feasible parsing batch size by converting the budget into a maximum number of parsable new papers:
$
\bar{\ell}(\Delta \repo_t)
=
\frac{1}{|\Delta \repo_t|}\sum_{d\in \Delta \repo_t}\ell(d),
\;
M_t
=
\min\!\left\{|\Delta \repo_t|,\;
\left\lfloor \frac{K_t}{\bar{\ell}(\Delta \repo_t)} \right\rfloor
\right\}.
$
If we plan to parse $M_t$ papers in this iteration, then the expected total token consumption becomes
$
M_t\sum_{h=1}^{H_{t-1}}\Pi_{t,h}\bar{\ell}_h,
$
and budget feasibility requires
$
M_t\sum_{h=1}^{H_{t-1}}\Pi_{t,h}\bar{\ell}_h \le K_t.
$

\noindent\textbf{Optimization Formulation.}
Combining these components, we have the following constrained maximization problem:
\begin{equation}
\label{eq:planning_obj}
\max_{\Pi_t}\; 
\sum_{h=1}^{H_{t-1}} \Pi_{t,h}\,U_h
\;-\;
\lambda_{\mathrm{plan}}\,\mathrm{KL}\big(\Pi_t\,\|\,\Pi_{t-1}^{\mathrm{prior}}\big) \;\;
\textrm{s.t.}\;
M_t\sum_{h=1}^{H_{t-1}}\Pi_{t,h}\,\bar{\ell}_h \;\le\; K_t,
\end{equation}
where $\lambda_{\mathrm{plan}}>0$ controls the trade-off between maximizing empirical evidence gain and adhering to the context-realization prior.

\subsubsection{Gaussian Proxy Utility and Cost Estimation.}
\label{sec:gaussian_proxy}
The key challenge in Eq.~\ref{eq:planning_obj} is to instantiate, for each candidate context $\mathcal{C}_h$,
(i) an evidence utility $U_h$ and (ii) an expected per-paper parsing cost $\bar{\ell}_h$, in a way that is both principled and computationally tractable.

\noindent\textbf{Gaussian-Proxy Utility.}
Recall that each context $\mathcal{C}_h$ induces a hypothesis distribution over candidate boxes in the learned geometric space,
$\pi_h(B)\triangleq P(B\mid \mathcal{C}_h,\mathcal{T}_{t-1})$.
Drawing $B\sim \pi_h$ and mapping it through the learned box-to-word function $\mathbf{x}=\btow(B)$ induces a (generally intractable)
distribution over query embeddings.
We score a context by its \emph{expected} relevance to the newly arrived papers:
\begin{equation}\label{eq:utility_general}
U_h \triangleq \mathbb{E}_{d\sim \hat{p}_t}\big[r_h(d)\big],
\;\;
r_h(d)\triangleq \mathbb{E}_{B\sim \pi_h}\Big[k\big(\btow(B),\mathbf{s}(d)\big)\Big],
\end{equation}
where $\hat{p}_t$ is uniform over $\Delta\repo_t$, $\mathbf{s}(d)$ denotes the embedding of paper $d$, and $k(\cdot,\cdot)$ is a similarity kernel in the word embedding space.

We adopt three design choices to make Eq.~\ref{eq:utility_general} tractable and scalable.
\emph{(i) Gaussian proxy.} We approximate the induced query distribution by a Gaussian,
$\mathbf{x}_h \approx \mathcal{N}(\mu_h,\Sigma_h)$, with $(\mu_h,\Sigma_h)$ estimated via a first-order Delta approximation
(i.e., moment propagation through $\btow(\cdot)$)~\cite{oehlert1992note}. This yields a compact, low-variance representation of the
context-induced uncertainty and enables closed-form scoring with common kernels.
\emph{(ii) RBF kernel.} We use an RBF kernel $k(\cdot,\cdot)$ as a smooth, bounded, distance-based relevance measure; under a Gaussian
proxy, its expectation admits a closed form~\cite{bishop2006prml}, avoiding expensive per-context Monte Carlo over $B$.
\emph{(iii) Sampling.} Computing $U_h$ exactly requires scanning all $|\Delta\repo_t|$ papers. We instead estimate the outer expectation using $m$ i.i.d.\ samples, achieving sublinear runtime with a tunable accuracy--efficiency trade-off.

\noindent\textbf{Query-specific Cost Estimation.}
To incorporate token costs, we treat $r_h(d)$ as a soft relevance weight and define a query-induced distribution over papers
$
p(d\mid h)\;\propto\; \hat{p}_t(d)\,r_h(d).
$
Let $\ell(d)$ denote the token length of the related-work section of paper $d$.
We define the expected per-paper parsing cost induced by context $\mathcal{C}_h$ as
$
\label{eq:cost-def}
\bar{\ell}_h
=
\mathbb{E}_{d\sim p(d\mid h)}[\ell(d)]
=
\mathbb{E}_{d\sim \hat{p}_t}\!\big[r_h(d)\,\ell(d)\big] \Big/ \mathbb{E}_{d\sim \hat{p}_t}\!\big[r_h(d)\big].
$
Using the same sampled papers $\{d_j\}_{j=1}^{m}$, we estimate $\bar{\ell}_h$ efficiently by
$
\widehat{\bar{\ell}}_h
=
\sum_{j=1}^{m} r_h(d_j)\,\ell(d_j) \Big/ \sum_{j=1}^{m} r_h(d_j).
$

Based on the above, Eq.~\ref{eq:planning_obj} admits a unique closed-form solution, summarized below, with a full proof provided in Appendix~\ref{sec:proof_of_constrained_problem}.

\begin{proposition}[Optimal Budget-Aware KL-Regularized Policy]\label{pro:exploration_policy}
For Eq.~\ref{eq:planning_obj}, there exists $\lambda^*\ge 0$ such that the unique maximizer satisfies
$
\Pi_{t,h}^{*}=\frac{\Pi_{t-1,h}^{\mathrm{prior}}\exp\!\big((U_h-\lambda^* \bar{\ell}_h)/\lambda_{\mathrm{plan}}\big)}{Z(\lambda^*)},
$
where
$
Z(\lambda)=\sum_{h'=1}^{H_{t-1}}\Pi_{t-1,h'}^{\mathrm{prior}}\exp\!\big((U_{h'}-\lambda \bar{\ell}_{h'})/\lambda_{\mathrm{plan}}\big).
$
The optimal $\lambda^*$ can be found by one-dimensional bisection to satisfy the budget constraint (with equality when active):
$
M_t\sum_{h=1}^{H_{t-1}} \Pi_{t,h}^{*}\bar{\ell}_h \le K_t.
$
\end{proposition}

\subsection{Plan Materialization and Paper Selection}
\label{sec:execution}
With the optimal exploration plan $\Pi_t$ established, we proceed to materialize and execute it. This involves determining the number of hypothetical boxes to sample per direction, transforming them into query embeddings to initialize the search, and selecting the papers that most effectively refine the taxonomy.

\noindent\textbf{Exploration Strategy Materialization.}
We convert the fractional plan $\Pi_t^{*}$ into integer per-context counts by allocating $\tilde n_h=\Pi_{t,h}^{*}M_t$ and choosing integers $n_h\in\mathbb{Z}_{\ge 0}$ with $\sum_{h=1}^{H_{t-1}} n_h = M_t$, obtained via unbiased rounding with a largest-remainder tie-break.

\noindent\textbf{Query Sampling and Retrieval.}
For each context $\mathcal{C}_h$, draw $n_h$ box samples $B$ and map each to a query embedding $\mathbf{x}=\btow(B)$, forming the query set $\mathcal{X}_t$.
For each $\mathbf{x}\in\mathcal{X}_t$, retrieve its top-$L$ nearest papers from the \emph{unparsed index} $\mathcal{R}_t$---the embedding index over papers that have arrived by iteration $t$ but have not yet been parsed---and let $\repo_t^{\mathrm{cand}}\subseteq\mathcal{R}_t$ be the deduplicated union of these candidates.
Let $k(\cdot,\cdot)\in[0,1]$ be the similarity and define $a(d,\mathbf{x})=k(\mathbf{s}(d),\mathbf{x})$ for $d\in\repo_t^{\mathrm{cand}}$, $\mathbf{x}\in\mathcal{X}_t$.

\noindent\textit{Submodular objective.}
We select a subset $S\subseteq\repo_t^{\mathrm{cand}}$ by solving the knapsack-constrained facility-location problem:
\begin{equation}\label{eq:evidence_retrieval}
\max_{S\subseteq\repo_t^{\mathrm{cand}}}\ \ F(S)\ =\ \sum_{\mathbf{x}\in\mathcal{X}_t}\ \max_{d\in S} a(d,\mathbf{x})
\;\;\;\text{s.t.}\;
\sum_{d\in S}\ell(d)\ \le\ K_t,
\end{equation}
Eq.~\ref{eq:evidence_retrieval} is a monotone submodular maximization under a single knapsack constraint, for which the proof is provided in Appendix~\ref{sec:proof_submodularity_f2}. Similarly, we solve it using a standard cost-aware lazy-greedy procedure with partial enumeration.
Similarly, we solve it using a standard cost-aware lazy-greedy procedure with partial enumeration.
Note that Eq.~\ref{eq:evidence_retrieval} does not introduce an additional budget: it enforces the same per-iteration cap $K_t$ on the \emph{realized} selection $S$, whereas the constraint in Eq.~\ref{eq:planning_obj} is its planning-stage counterpart that bounds the \emph{expected} parsing cost when shaping the allocation $\Pi_t$.

\noindent\textbf{LLM-guided Partial Taxonomy Extraction.}
The extraction process first employs regular-expression--based matching to isolate the \emph{Related Work} sections.
Subsequently, \method prompts an LLM to extract partial taxonomies.
Detailed regex matching process and prompt specifications are provided in Appendix~\ref{sec:details_for_LLM_extractor}.

\section{Incremental Semantic Indexing}\label{sec:incremental_index}
At iteration $t$, we maintain a taxonomy-induced inverted index $\mathcal{I}_t$ to support efficient downstream retrieval. Newly arrived papers $\Delta \mathcal{P}_t$ are indexed immediately by assigning them to relevant concept nodes via embedding similarity. For previously indexed papers, we condition updates on the refinement type to avoid full-corpus re-indexing:
(1) \textit{Expansion.} Introducing a new node $v$ requires no historical re-indexing: $\mathcal{I}_t[v]$ is initialized as empty and populated only by matches from $\Delta \mathcal{P}_t$ and future streams.
(2) \textit{Specialization.} Specializing a node $u$ into children $\mathcal{S}_u$ requires redistributing only the affected posting list $\mathcal{I}_{t-1}[u]$ across $\mathcal{S}_u$, avoiding any full-repository scan.
(3) \textit{Abstraction.} Creating a new parent node $a$ over an existing child set $\mathcal{S}_a$ requires only posting-list aggregation:
$
\mathcal{I}_t[a] \leftarrow \bigcup_{c\in\mathcal{S}_a} \mathcal{I}_{t-1}[c].
$
This is a local operation requiring no historical re-scoring.

\noindent\textbf{Supporting Paper Retrieval.} Given a user query $q$, we route it to the taxonomy by selecting the top-$k$ semantically similar concepts $\mathcal{U}_q$. We then form a compact candidate set by unioning the posting lists of these concepts and their descendants:
$
\repo_t(q)=\bigcup_{u\in\mathcal{U}_q}\left(\mathcal{I}_t[u]\ \cup\ \bigcup_{v\in \mathrm{Desc}(u)} \mathcal{I}_t[v]\right).
$
Finally, papers in $\repo_t(q)$ are ranked using standard retrieval models (e.g., BM25 or dense retrieval), enabling efficient candidate generation without repetitive full-corpus scanning.

\section{Experiments}\label{sec:exp}
In this section, we conduct experiments to evaluate: 1) the effectiveness and efficiency of our proposed method on the \problem task (Sec.~\ref{sec:main_results}); 2) the effectiveness and efficiency of the semantic indexing based on the generated taxonomy (Sec.~\ref{sec:effectiveness_downstream}); 3) the sensitivity of key hyper-parameters (Sec.~\ref{sec:hyperparameter_study}); 4) the ablation study (Sec.~\ref{sec:sharper_ablation}); 5) diagnostic studies on extraction robustness, numerical stability, and predictive precision (Sec.~\ref{sec:robustness_studies}); and 6) a qualitative case study (Sec.~\ref{sec:case_llm2sql}).

\subsection{Evaluation of Taxonomy Maintenance}
\label{sec:main_results}
\noindent\textbf{Benchmarks.} 
We simulate an ``in-the-wild'' setting using the Computer Science (CS) subset of the arXiv corpus~\cite{ar5iv-dataset}, which comprises over 270K LaTeXML-converted HTML documents. 
To establish ground truth, we curate a benchmark of 12 survey papers spanning diverse CS domains. We select survey papers that are either published in top-tier venues or highly cited, manually extracting and verifying their author-curated taxonomies. 
To capture the evolutionary dynamics of each topic, we assign distinct start times to each domain, ensuring that the topic has reached a sufficient level of development prior to the simulation. We then partition the remaining timeline into seven equal intervals $t=1,\dots,7$, where sequential paper arrivals trigger iterative taxonomy updates; the first interval supplies the initial corpus snapshot ($\repo_0$ in Definition~\ref{def:problem}), so the initial construction (\emph{Init}) is evaluated at $t{=}1$.
To avoid data leakage, we remove the benchmark survey papers from the arXiv corpus.
Benchmark statistics are reported in Table~\ref{tab:benchmark}, and further details on the survey papers are provided in Appendix~\ref{sec:details_for_experiment_setting}.

\noindent\textit{Remark.} We adopt a single highly authoritative survey per domain, selected from top-tier venues or with high citation counts, rather than aggregating across multiple surveys. Aggregation would require either redefining the evaluation protocol against an ensemble of ground truths or manually reconciling disparate hierarchies into a unified gold standard, both of which would introduce substantial curator-induced bias. We acknowledge that survey-derived taxonomies favor cleaner, consensus-driven hierarchies over fully organic ones; however, they capture the structured knowledge that downstream users and systems actually rely on, and purely organic benchmarks at scale are notoriously difficult to obtain. We leave the construction of such benchmarks as a promising direction for future work.

\noindent\textbf{Baselines.}
To comprehensively evaluate \method in the ``in-the-wild'' setting, we design baselines through a modular pipeline that factorizes:
(i) a \emph{taxonomy construction backbone}, (ii) an \emph{updating mechanism} for sequential arrivals, and (iii) a \emph{paper search strategy} for selecting evidence papers from the newly arrived papers. 
We denote each baseline instance as $\textsc{Backbone}_{\textsc{Update},\textsc{Search}}$.

\begin{table}[t]
\centering
\small
\setlength{\tabcolsep}{6pt}

\renewcommand{\arraystretch}{1.1} 

\setlength{\aboverulesep}{0pt}
\setlength{\belowrulesep}{0pt}

\caption{Statistics of the benchmark.}
\footnotesize
\label{tab:benchmark}
\begin{tabularx}{\linewidth}{lXcc}
\toprule
\textbf{Domain} & \textbf{Research Topic (Input)} & \textbf{Levels} & \textbf{Nodes} \\
\midrule
\multirow{5}{*}{DB}
& Schema Design \& Tuning (SDT) & 4 & 24 \\
& LLM-based Text-to-SQL (LLM2SQL) & 5 & 37 \\
& HTAP Databases (HTAP) & 2 & 28 \\
& Graph ANN Vector Search (ANN-VS) & 3 & 21 \\
& Distributed Graph Algorithms (DGA) & 3 & 31 \\
\midrule
\multirow{2}{*}{SE}
& Auto Program Repair (APR) & 3 & 14 \\
& AI for Software Testing (AI4ST) & 3 & 48 \\
\midrule
\multirow{2}{*}{ML}
& Machine Unlearning (MU) & 3 & 16 \\
& Fairness in ML (FairML) & 3 & 19 \\
\midrule
IR
& Multimodal Recommenders (MRS) & 3 & 29 \\
\midrule
HCI
& Industrial Extended Reality (I-XR) & 3 & 58 \\
\midrule
Robotics
& Multi-Robot Task Allocation (MRTA) & 5 & 23 \\
\bottomrule
\end{tabularx}
\end{table}

\begin{table}[t]
\scriptsize
\setlength{\tabcolsep}{1.8pt}
\caption{Taxonomy quality and efficiency comparison across all compared methods. Each row \textsc{U}+\textsc{A} under a backbone denotes the variant $\textsc{Backbone}_{\textsc{U},\textsc{A}}$, where \textsc{U} $\in$ \{S, I\} is the updating mechanism and \textsc{A} $\in$ \{K, C, H\} the evidence acquisition strategy. F1 metrics are reported as mean$_{\pm\text{std}}$ over 5 random seeds. $\dagger$ denotes a statistically significant improvement of our method GIST$_{I,H}$ over \emph{all baselines} under paired $t$-tests at $p<0.01$.}
\label{tab:main_results}
\begin{tabular}{l@{\hspace{4pt}}ccccccccc}
\toprule
 & \multicolumn{3}{c}{Node Soft F1} & \multicolumn{3}{c}{Edge Soft F1} & Runtime & Tokens & Cost \\
\cmidrule(lr){2-4}\cmidrule(lr){5-7}
 & Init & Final & $\Delta$/ep & Init & Final & $\Delta$/ep & (min) & (k) & (\$) \\
\midrule
\multicolumn{10}{l}{\emph{\taxogen}} \\
S+K & 24.8$_{\pm 0.3}$ & 31.2$_{\pm 0.4}$ & +1.07$_{\pm 0.05}$ & 19.7$_{\pm 0.2}$ & 24.7$_{\pm 0.3}$ & +0.83$_{\pm 0.04}$ & 76.4 & -- & -- \\
S+C & 24.8$_{\pm 0.3}$ & 33.8$_{\pm 0.5}$ & +1.50$_{\pm 0.06}$ & 19.7$_{\pm 0.2}$ & 27.1$_{\pm 0.4}$ & +1.23$_{\pm 0.05}$ & 77.8 & -- & -- \\
I+K & 24.8$_{\pm 0.3}$ & 29.4$_{\pm 0.3}$ & +0.77$_{\pm 0.04}$ & 19.7$_{\pm 0.2}$ & 23.6$_{\pm 0.3}$ & +0.65$_{\pm 0.04}$ & \textbf{22.7} & -- & -- \\
I+C & 24.8$_{\pm 0.3}$ & 31.7$_{\pm 0.4}$ & +1.15$_{\pm 0.05}$ & 19.7$_{\pm 0.2}$ & 25.4$_{\pm 0.3}$ & +0.95$_{\pm 0.04}$ & 23.6 & -- & -- \\
\midrule
\multicolumn{10}{l}{\emph{\hyperexpan}} \\
S+K & 32.8$_{\pm 0.4}$ & 40.7$_{\pm 0.5}$ & +1.32$_{\pm 0.06}$ & 26.7$_{\pm 0.3}$ & 33.8$_{\pm 0.4}$ & +1.18$_{\pm 0.05}$ & 91.2 & -- & -- \\
S+C & 32.8$_{\pm 0.4}$ & 43.6$_{\pm 0.6}$ & +1.80$_{\pm 0.07}$ & 26.7$_{\pm 0.3}$ & 36.4$_{\pm 0.5}$ & +1.62$_{\pm 0.06}$ & 92.7 & -- & -- \\
I+K & 32.8$_{\pm 0.4}$ & 38.7$_{\pm 0.4}$ & +0.98$_{\pm 0.05}$ & 26.7$_{\pm 0.3}$ & 32.4$_{\pm 0.3}$ & +0.95$_{\pm 0.04}$ & 27.1 & -- & -- \\
I+C & 32.8$_{\pm 0.4}$ & 41.2$_{\pm 0.5}$ & +1.40$_{\pm 0.06}$ & 26.7$_{\pm 0.3}$ & 34.8$_{\pm 0.4}$ & +1.35$_{\pm 0.05}$ & 28.4 & -- & -- \\
\midrule
\multicolumn{10}{l}{\emph{\taxoalign}} \\
S+K & 38.6$_{\pm 0.7}$ & 56.2$_{\pm 0.8}$ & +2.93$_{\pm 0.13}$ & 32.3$_{\pm 0.6}$ & 47.8$_{\pm 0.7}$ & +2.58$_{\pm 0.11}$ & 451.8 & 2438 & 7.31 \\
S+C & 38.6$_{\pm 0.7}$ & 60.4$_{\pm 0.9}$ & +3.63$_{\pm 0.14}$ & 32.3$_{\pm 0.6}$ & 51.3$_{\pm 0.8}$ & +3.17$_{\pm 0.12}$ & 458.4 & 2552 & 7.64 \\
I+K & 38.6$_{\pm 0.7}$ & 46.4$_{\pm 0.6}$ & +1.30$_{\pm 0.10}$ & 32.3$_{\pm 0.6}$ & 39.7$_{\pm 0.5}$ & +1.23$_{\pm 0.09}$ & 113.4 & 871 & 2.62 \\
I+C & 38.6$_{\pm 0.7}$ & 49.8$_{\pm 0.7}$ & +1.87$_{\pm 0.11}$ & 32.3$_{\pm 0.6}$ & 42.6$_{\pm 0.6}$ & +1.72$_{\pm 0.10}$ & 117.2 & 928 & 2.78 \\
\midrule
\multicolumn{10}{l}{\emph{\taxoadapt}} \\
S+K & 42.8$_{\pm 0.8}$ & 61.7$_{\pm 0.9}$ & +3.15$_{\pm 0.13}$ & 35.8$_{\pm 0.7}$ & 53.6$_{\pm 0.8}$ & +2.97$_{\pm 0.12}$ & 384.6 & 3756 & 11.27 \\
S+C & 42.8$_{\pm 0.8}$ & 66.2$_{\pm 0.9}$ & +3.90$_{\pm 0.14}$ & 35.8$_{\pm 0.7}$ & 57.8$_{\pm 0.8}$ & +3.67$_{\pm 0.12}$ & 391.8 & 3895 & 11.69 \\
I+K & 42.8$_{\pm 0.8}$ & 52.4$_{\pm 0.7}$ & +1.60$_{\pm 0.10}$ & 35.8$_{\pm 0.7}$ & 44.1$_{\pm 0.6}$ & +1.38$_{\pm 0.09}$ & 114.3 & 947 & 2.84 \\
I+C & 42.8$_{\pm 0.8}$ & 56.7$_{\pm 0.8}$ & +2.32$_{\pm 0.12}$ & 35.8$_{\pm 0.7}$ & 47.6$_{\pm 0.7}$ & +1.97$_{\pm 0.10}$ & 117.3 & 1003 & 3.01 \\
\midrule
\multicolumn{10}{l}{\emph{\method}} \\
S+K & \textbf{47.5}$_{\pm 0.5}$ & 58.6$_{\pm 0.5}$ & +1.85$_{\pm 0.07}$ & \textbf{40.8}$_{\pm 0.4}$ & 50.4$_{\pm 0.4}$ & +1.60$_{\pm 0.06}$ & 113.4 & 641 & 1.91 \\
S+C & \textbf{47.5}$_{\pm 0.5}$ & 62.3$_{\pm 0.6}$ & +2.47$_{\pm 0.08}$ & \textbf{40.8}$_{\pm 0.4}$ & 53.7$_{\pm 0.5}$ & +2.15$_{\pm 0.07}$ & 114.6 & 682 & 2.04 \\
S+H & \textbf{47.5}$_{\pm 0.5}$ & \textbf{76.0}$_{\pm 0.6}$ & \textbf{+4.75}$_{\pm 0.10}$ & \textbf{40.8}$_{\pm 0.4}$ & \textbf{67.6}$_{\pm 0.5}$ & \textbf{+4.47}$_{\pm 0.09}$ & 117.3 & 713 & 2.14 \\
I+K & \textbf{47.5}$_{\pm 0.5}$ & 55.3$_{\pm 0.4}$ & +1.30$_{\pm 0.06}$ & \textbf{40.8}$_{\pm 0.4}$ & 48.2$_{\pm 0.4}$ & +1.23$_{\pm 0.06}$ & 35.8 & \textbf{421} & \textbf{1.26} \\
I+C & \textbf{47.5}$_{\pm 0.5}$ & 60.5$_{\pm 0.5}$ & +2.17$_{\pm 0.07}$ & \textbf{40.8}$_{\pm 0.4}$ & 53.1$_{\pm 0.4}$ & +2.05$_{\pm 0.06}$ & 36.7 & 463 & 1.38 \\
I+H & \textbf{47.5}$^{\dagger}_{\pm 0.5}$ & 73.5$^{\dagger}_{\pm 0.5}$ & +4.33$^{\dagger}_{\pm 0.10}$ & \textbf{40.8}$^{\dagger}_{\pm 0.4}$ & 65.4$^{\dagger}_{\pm 0.4}$ & +4.10$^{\dagger}_{\pm 0.09}$ & 37.5 & 498 & 1.49 \\
\bottomrule
\end{tabular}
\end{table}

\noindent\emph{\underline{Backbones.}}
We instantiate four representative baselines spanning the major taxonomy-construction paradigms: 1)~\taxogen~\cite{Zhang2018TaxoGen}, a point-embedding clustering method that iteratively partitions a topic term graph; 2)~\hyperexpan~\cite{ma2021hyperexpan}, a hyperbolic-embedding expansion method that places new concepts in negatively-curved space and attaches them via geometric distances; 3)~\taxoalign~\cite{lahiri2025taxoalign} builds taxonomies via a three-stage pipeline: locating topic-related knowledge slices, aggregating these slices to form an initial taxonomy, and performing LLM-based refinement; and 4)~\taxoadapt~\cite{kargupta2025taxoadapt} starts with an LLM-generated taxonomy and iteratively \emph{grounds} it to the collected documents. It employs top-down classification and selective expansion, driven by corpus density and unmapped signals, to generate a final corpus-aligned taxonomy. To provide initial concept seeds for \taxogen and \hyperexpan, we use BERTopic~\cite{grootendorst2022bertopic} to extract them from the input corpus.

\noindent\emph{\underline{Updating mechanisms.}}
We instantiate each backbone under two standard updating mechanisms:
(1) \textit{Scratch (S):} Ignores the previous state and leverages the full cumulative corpus up to interval $t$.
For the baselines (\taxoalign and \taxoadapt), this involves reconstructing the taxonomy entirely from the cumulative set of papers.
For \method, this entails retraining the mapping-model parameters from scratch using all accumulated data.
(2) \textit{Incremental (I):} Updates the taxonomy $\mathcal{T}_{t-1} \rightarrow \mathcal{T}_t$ using only the newly arrived papers in interval $t$.
For the baselines, we implement a generate-and-merge approach: (i) running the backbone on the new batch to generate a ``delta'' taxonomy $\Delta\mathcal{T}_t$; and (ii) merging $\Delta\mathcal{T}_t$ with $\mathcal{T}_{t-1}$ via concept alignment and edge union.
For \method, we employ the proposed incremental training strategy (Sec.~\ref{sec:incremental_training}).

\noindent\emph{\underline{Evidence Acquisition strategies.}}
We consider three strategies: 
(1) \textit{Topic-keyword search (K)}: queries using the initial topic keywords as the query at each update; 
(2) \textit{Concept search (C)}: queries using concepts in the current taxonomy; 
and (3) \textit{Hypothesized-concept search (H)}: queries using the hypothesized concepts (specific to our method).
To ensure a fair comparison across methods that may rely on different parts of each paper, we cap the amount of raw text processed from newly retrieved papers at every update.
Specifically, the total evidence budget per iteration is limited to 400{,}000 tokens in the first iteration and 100{,}000 tokens thereafter, resulting in a 1{,}000{,}000-token budget per topic. We use GPT-5-mini for all methods.
Note that this budget caps only the \emph{evidence text supplied as initial LLM input} (i.e., the raw paper content parsed per update); the total token counts reported in Table~\ref{tab:main_results} additionally include each method's own prompting scaffolds, intermediate reasoning calls, and generated outputs, and can therefore exceed the evidence budget.

\noindent\textbf{Evaluation Metrics.} We evaluate three dimensions: effectiveness, efficiency, and monetary cost. For \emph{effectiveness}, we report the standard \nodefone and \edgefone metrics~\cite{lahiri2025taxoalign}. \nodefone quantifies concept-level agreement via semantic maximum-weight bipartite matching, while \edgefone measures agreement on parent--child relations, weighted by the similarity of their endpoints. 
At each iteration we evaluate against the final ground truth ($t=7$). \emph{Efficiency} is measured by total runtime across all iterations, and \emph{monetary cost} by total API expense. 

\noindent\textbf{Implementation and Hyper-parameters.}
We implement \method in PyTorch and run all experiments on a Linux server equipped with an NVIDIA RTX 4090 GPU. For geometric representation, we set the box-embedding dimension to $d=12$. In the geometric-guided taxonomy integration module, we set the reliability weight to $\lambda_{\mathrm{rel}}=0.3$ and the novelty weight to $\lambda_{\mathrm{nov}}=0.2$. For the \emph{Cost-effective Evidence Retrieval} module, we set the planning weight to $\lambda_{\mathrm{plan}}=0.2$. For Gaussian proxy utility estimation, we employ a sample size of $m=1000$, which keeps the sampling error of the estimated utilities negligible in practice. The source code is publicly available in~\cite{gist_code_repo}.

\noindent\textbf{Results}. The comparison results are presented in Table~\ref{tab:main_results}. We make the following observations:
\begin{itemize}[leftmargin=*,noitemsep,topsep=0pt]
\item \textbf{Overall.} The \method~series achieves the best taxonomy quality across every quality metric in the grid. The only baselines with lower runtime or monetary cost are the two embedding-only methods (\taxogen~and \hyperexpan), but their taxonomy quality is far below competitive levels (Final Node F1 below 44\%). Among methods with competitive quality, the proposed efficient variant \methodvar{I}{H} attains the best cost--quality trade-off: an 11.0\% / 13.1\% improvement in Final Node / Edge F1 over the strongest LLM baseline \adaptvar{S}{C}, statistically significant at $p<0.01$, while requiring only 9.6\% of its runtime and 12.7\% of its monetary cost. Note that \hyperexpan and \taxogen do not involve LLM processing and therefore incur no API cost and run efficiently. However, their performance lags substantially behind LLM-based methods, reflecting the inherent trade-off between cost and taxonomy quality. 

\item \textbf{Backbone.} Because the \emph{Init} score is measured at $t{=}1$ before any updating mechanism or evidence acquisition strategy takes effect, it isolates the contribution of the backbone itself. The four baseline backbones cluster into three quality tiers reflecting their representational power. The point-embedding backbone (\taxogen) cannot model the asymmetric containment relations required for is-a hierarchies and remains below 35\% \emph{Final} \nodefone~under any update or search configuration. The hyperbolic backbone (\hyperexpan) modestly improves this via negatively-curved geometry but still tops out around 44\%. The LLM-driven backbones (\taxoalign~and \taxoadapt) reach 60--66\% by leveraging LLM reasoning, yet exhibit the high token cost and structural hallucination characteristic of free-form generation. \method's box-embedding backbone, in contrast, attains the highest \emph{Init} ($47.5$ Node, $40.8$ Edge) -- evidence that combining expert-curated partial taxonomies with a box-embedding space, which natively encodes containment, yields a stronger structural prior than any of the alternative backbones.

\item \textbf{Updating Mechanism.} An effective updating mechanism is essential for the task of \problem. Without one (i.e., retraining from scratch at each interval), runtime grows substantially for every method, and monetary cost grows dramatically for the LLM-intensive backbones \taxoadapt~and \taxoalign.
With our dedicated updating mechanism and novelty-aware coreset selection, \method's incremental variant \methodvar{I}{H} retains $96.7\%$ of the fully retrained \methodvar{S}{H}'s quality (Final Node F1: $73.5$ vs.\ $76.0$; Final Edge F1: $65.4$ vs.\ $67.6$, a $3.3\%$ relative drop) while cutting runtime from $117.3$ to $37.5$ minutes ($3.1\times$ speedup) and monetary cost from \$2.14 to \$1.49 ($30.4\%$ reduction). This demonstrates that scratch retraining is largely unnecessary when guided by the hybrid evidence strategy.
By contrast, naive incremental updating in \taxoadapt~and \taxoalign~reduces runtime by $70.1\%$ / $74.4\%$ and monetary cost by $74.2\%$ / $63.6\%$, but at the expense of a $14.4\%$ / $17.5\%$ relative quality drop. The gap between \method's $3.3\%$ drop and the $14$--$18\%$ drop suffered by these baselines shows that simple merging is insufficient for preserving structural integrity, whereas \method's coreset-based mechanism keeps nearly all of the full-retraining quality at a fraction of the cost.

\item \textbf{Evidence Acquisition Strategy.} The evidence acquisition strategy via \emph{hypothesized concepts} drives the fastest taxonomy growth per epoch. It enables more targeted information acquisition and avoids retrieving redundant papers that contribute only to already-known concepts. Compared with leveraging existing \emph{concepts} and \emph{input search topics}, \method's per-epoch gain ($\Delta$/ep) is $2.0\times$ and $3.3\times$ higher on Node Soft F1 (\methodvar{I}{H}: $+4.33$, \methodvar{I}{C}: $+2.17$, \methodvar{I}{K}: $+1.30$), with the same $2.0\times$ / $3.3\times$ gap on Edge ($+4.10$ vs.\ $+2.05$ / $+1.23$).
Querying with \emph{concepts} from the current taxonomy is also more likely to surface papers capable of expanding the taxonomy than querying with the original research topics alone. Conversely, relying on input keywords yields the smallest per-epoch progress among LLM-based incremental variants.
\end{itemize}

\subsection{Effectiveness in Downstream Literature Search}\label{sec:effectiveness_downstream}
We evaluate the effectiveness and efficiency of the proposed semantic indexing utilizing the maintained taxonomies. We reuse the arXiv corpus and the benchmark, with the following modifications.

\noindent\textbf{Query and Ground-Truth Construction.}
For each topic, we treat each concept node as an information need and generate two query types: (i) \emph{keyword queries} using the concept name, and (ii) \emph{natural-language queries} by prompting an LLM to produce a concise, user-style search query from the concept name and description; we generate 5 queries of each type per topic. We derive ground truth from the corresponding survey paper: for each concept query, the relevant set comprises papers explicitly assigned to that concept/category by the survey. We ensure these papers are present in the repository by adding any missing ones to the arXiv corpus.

\begin{table}[t]
\small
\setlength{\tabcolsep}{2pt}
\caption{Effectiveness comparison on paper search tasks. GIST and GIST$^{*}$ maintain the same taxonomy and produce identical indexes, hence identical retrieval effectiveness; The symbols $\dagger$ and $\ddagger$ indicate that GIST significantly outperforms all non-GIST baselines under paired $t$-tests at $p<0.01$ and $p<0.05$, respectively.}
\label{tab:paper_search}
\begin{tabular}{l|cc|cc|cc}
\hline
\multirow{2}{*}{Method} & \multicolumn{2}{c|}{BM25} & \multicolumn{2}{c|}{Dense} & \multicolumn{2}{c}{Hybrid} \\ \cline{2-7}
                        & nDCG@10 & MRR & nDCG@10 & MRR & nDCG@10 & MRR \\ \hline
Flat-Ret      & 0.374 & 0.333 & 0.402 & 0.356 & 0.438 & 0.378 \\
\taxoalign    & 0.382 & 0.339 & 0.409 & 0.363 & 0.448 & 0.387 \\
\taxoadapt    & 0.397 & 0.350 & 0.423 & 0.373 & 0.463 & 0.401 \\
\method{} / GIST$^{*}$ & \textbf{0.409}$^{\dagger}$ & \textbf{0.360}$^{\dagger}$ & \textbf{0.439}$^{\dagger}$ & \textbf{0.386}$^{\dagger}$ & \textbf{0.479}$^{\ddagger}$ & \textbf{0.420}$^{\dagger}$ \\ \hline
\end{tabular}
\end{table}

\begin{table}[t]
\small
\setlength{\tabcolsep}{3pt}
\caption{Evaluation of indexing efficiency. Index (MB), Build (min), and Per-paper (ms) count only the \emph{additional} taxonomy-induced index (posting lists + metadata) built on top of the base BM25/dense retrieval indexes shared by all methods.}
\label{tab:index_efficiency}
\begin{tabular}{lcccc}
\toprule
Method & Index (MB) & Build (min) & Per-paper (ms) & Query (ms) \\
\midrule
Flat-Ret      & 0.00          & 0.0          & 0.0            & 58.0 \\
\taxoalign    & 2.76          & 6.5          & 82.3           & 45.7 \\
\taxoadapt    & 3.15          & 7.4          & 91.8           & \textbf{41.9} \\
\methods      & \textbf{2.64} & 6.8          & 86.5           & 46.7 \\
\method       & \textbf{2.64} & \textbf{3.4} & \textbf{41.2}  & 45.4 \\
\bottomrule
\end{tabular}
\end{table}

\noindent\textbf{Evaluation Metrics.}
We report standard ranked-retrieval metrics averaged over topics, including \emph{nDCG@10} and \emph{MRR}, under three retrieval paradigms: \emph{BM25}, \emph{Dense}, and \emph{Hybrid}.
To quantify indexing efficiency, we additionally report:
(i) \emph{index size} of the taxonomy-induced inverted index at $t{=}7$ (posting lists + metadata),
(ii) cumulative \emph{index build/update time} across all iterations, counting only the time to materialize and maintain the index given the constructed taxonomy (i.e., assigning papers to concept nodes and updating posting lists), and
(iii) \emph{query latency}.

\noindent\textbf{Baselines.} We compare \method against taxonomies produced by \taxoalign and \taxoadapt (using their best-performing variants), and a flat retrieval baseline (\emph{Flat-Ret}) that performs retrieval directly over the corpus without any taxonomy-induced semantic index.
To isolate the benefit of incremental indexing, we further include $\method^{*}$, which uses the same taxonomy as \method but \emph{rebuilds} the semantic index at each iteration via full-corpus scanning.

\noindent\textbf{Results}.
Table~\ref{tab:paper_search} reports downstream search effectiveness. We observe that all taxonomy-indexed methods consistently outperform the \emph{Flat-Ret} baseline across BM25, Dense, and Hybrid retrieval, validating that a structured taxonomy improves semantic access to scholarly corpora. Moreover, \method achieves the best effectiveness overall. Compared with the strongest baseline \taxoadapt, \method yields consistent gains; averaging nDCG@10 and MRR, it improves performance by 2.9\%, 3.6\%, and 4.1\% under BM25, Dense, and Hybrid, respectively.
Table~\ref{tab:index_efficiency} reports indexing efficiency under two complementary views: cumulative \emph{build time} across the 7 timestamps, and the \emph{per-paper marginal cost} that characterizes streaming maintenance. \method halves the cumulative build time of $\method^{*}$ (full-scan indexing, 3.4 vs.\ 6.8 minutes) while preserving the same index size (2.64 MB), comparable query latency (45.4 vs.\ 46.7 ms), and identical retrieval effectiveness. Decomposed at the per-paper level, the marginal indexing cost of \method is only 41.2 ms, roughly 2$\times$ faster than $\method^{*}$ (86.5 ms) and the rebuild-from-scratch baselines (\taxoalign and \taxoadapt at 82.3--91.8 ms). 
Together, these results demonstrate that the induced taxonomy can serve as an accurate index structure at light maintenance overhead.

\subsection{Hyper-parameter Study}\label{sec:hyperparameter_study}
We analyze the sensitivity of \method{} to the key coefficients $\lambda_{\mathrm{rel}}$ (Eq.~\ref{eq:cover_rel}), $\lambda_{\mathrm{nov}}$ (Eq.~\ref{eq:nov-kl}), and $\lambda_{\mathrm{plan}}$ (Eq.~\ref{eq:planning_obj}), as well as the box dimension $d$, reporting the final Node and Edge Soft F1 of \method averaged over the 12 benchmark topics (Fig.~\ref{fig:hparam}). We make the following observations: 1) The reliability weight $\lambda_{\mathrm{rel}}$ balances concept coverage against evidence quality: performance peaks at $\lambda_{\mathrm{rel}}{=}0.3$ and stays close to the peak within $[0.2,0.5]$, while values toward either extreme over-emphasize one side of the trade-off---coverage-only selection ($\lambda_{\mathrm{rel}}{=}0$) admits noisy relations, whereas overly reliability-heavy selection over-prunes coverage---and degrade taxonomy quality. 2) The novelty weight $\lambda_{\mathrm{nov}}$ controls the trade-off between loss preservation and diversity in coreset selection: a moderate setting ($\lambda_{\mathrm{nov}}{=}0.2$) works best. Disabling novelty entirely causes the sharpest drop in Fig.~\ref{fig:hparam}(b), as redundant coresets waste the update budget, whereas larger values sacrifice loss preservation for novelty and reduce coherence. 3) The planning weight $\lambda_{\mathrm{plan}}$ balances the context-realization prior against utility-driven exploitation: a moderate value ($\lambda_{\mathrm{plan}}{=}0.2$) is optimal. Smaller values cost only mildly, concentrating the budget myopically on currently high-utility contexts, while larger values over-regularize toward the prior and miss emerging gaps, degrading more sharply as the prior dominates. 4) The box dimension $d$ determines geometric expressiveness: small $d$ underfits the hierarchical structure, performance peaks around $d{=}12$, and larger dimensions gradually \emph{degrade} quality---consistent with overfitting the limited self-supervised signals---while raising computational cost, justifying our default of $d{=}12$.

\begin{figure}[t]
  \centering
  \captionsetup[subfigure]{font=small,skip=0pt}

  \includegraphics[width=0.35\columnwidth]{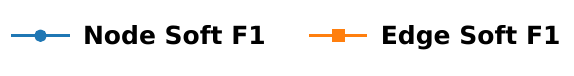}

  \begin{subfigure}[t]{0.24\columnwidth}
    \centering
    \includegraphics[width=\linewidth]{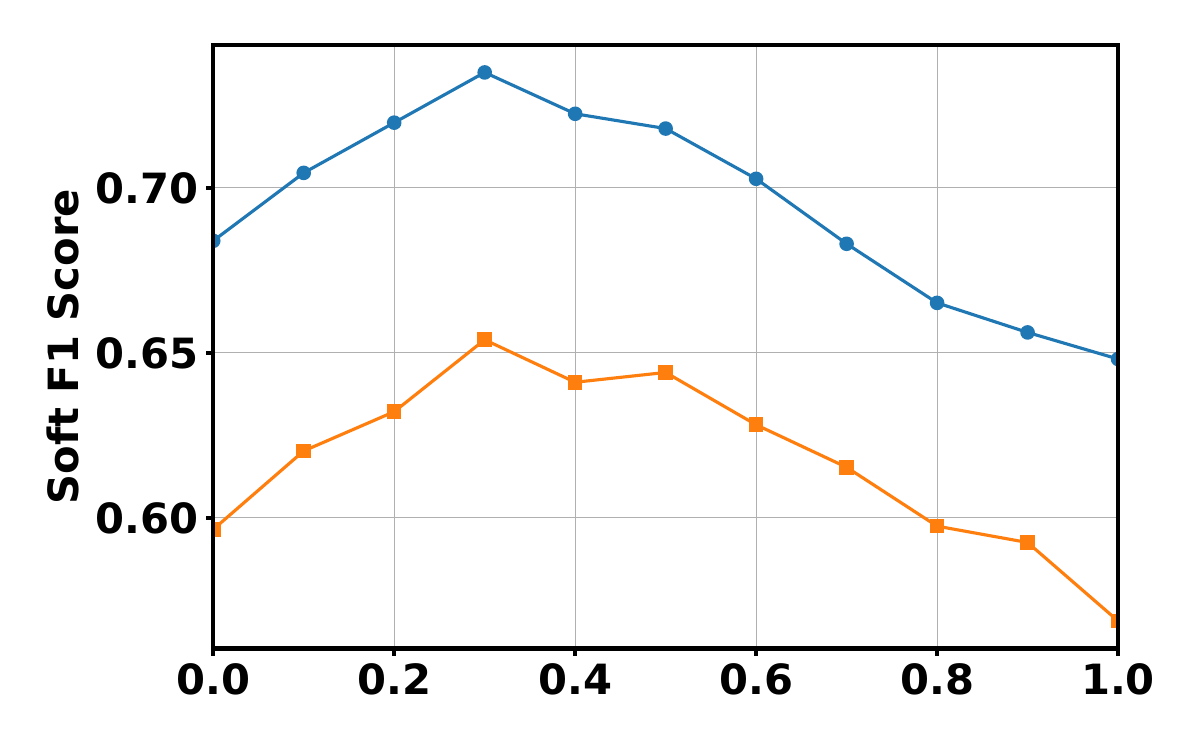}
    \caption{$\lambda_{\mathrm{rel}}$}
    \label{fig:hparam_rel}
  \end{subfigure}
  \hfill
  \begin{subfigure}[t]{0.24\columnwidth}
    \centering
    \includegraphics[width=\linewidth]{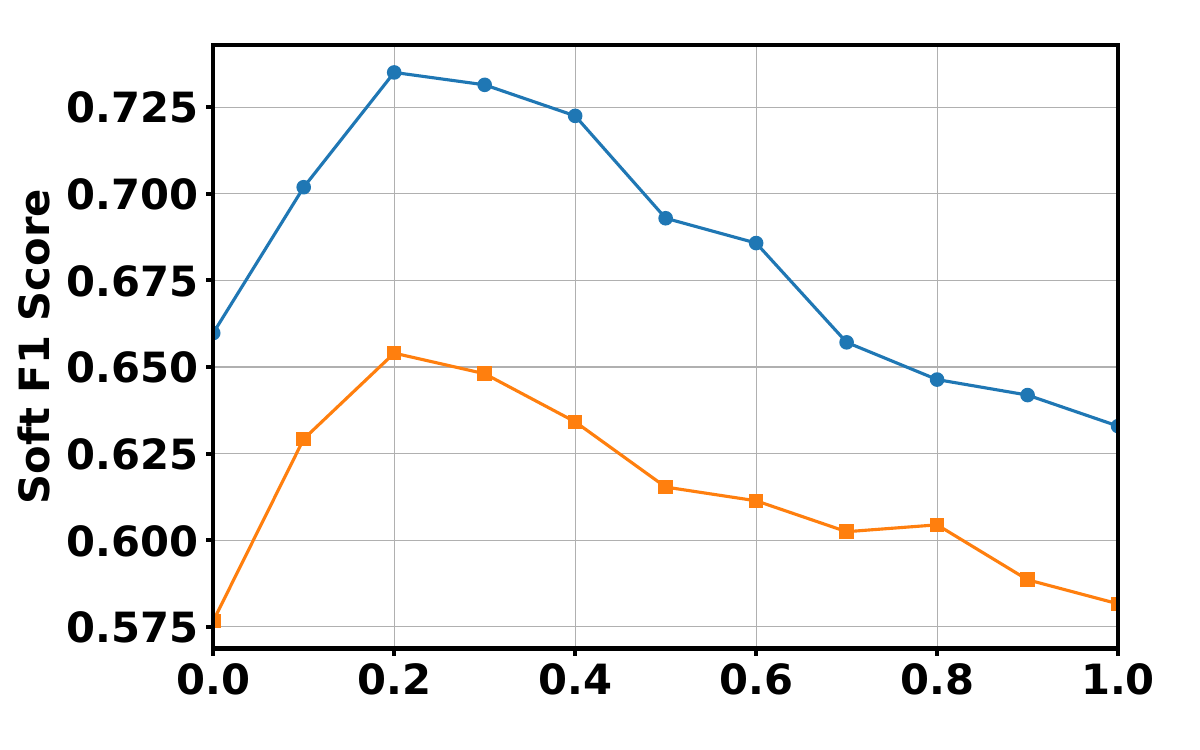}
    \caption{$\lambda_{\mathrm{nov}}$}
    \label{fig:hparam_nov}
  \end{subfigure}
  \hfill
  \begin{subfigure}[t]{0.24\columnwidth}
    \centering
    \includegraphics[width=\linewidth]{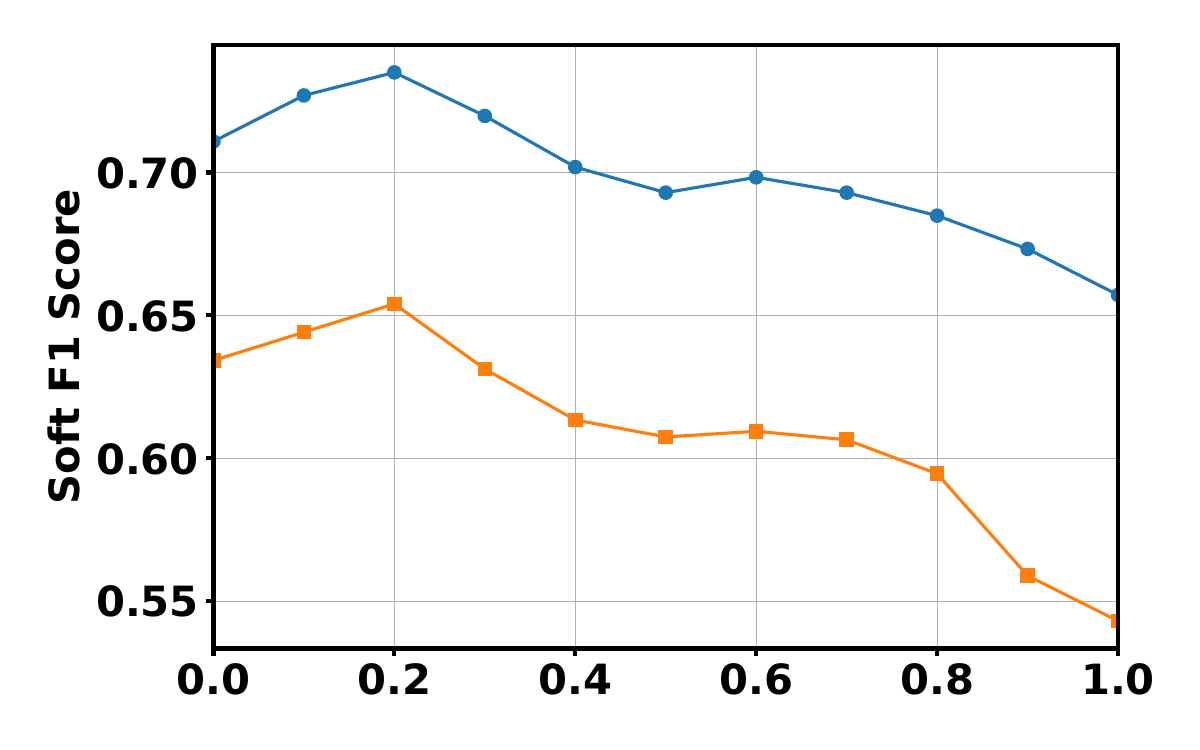}
    \caption{$\lambda_{\mathrm{plan}}$}
    \label{fig:hparam_plan}
  \end{subfigure}
  \hfill
  \begin{subfigure}[t]{0.245\columnwidth}
    \centering
    \includegraphics[width=\linewidth]{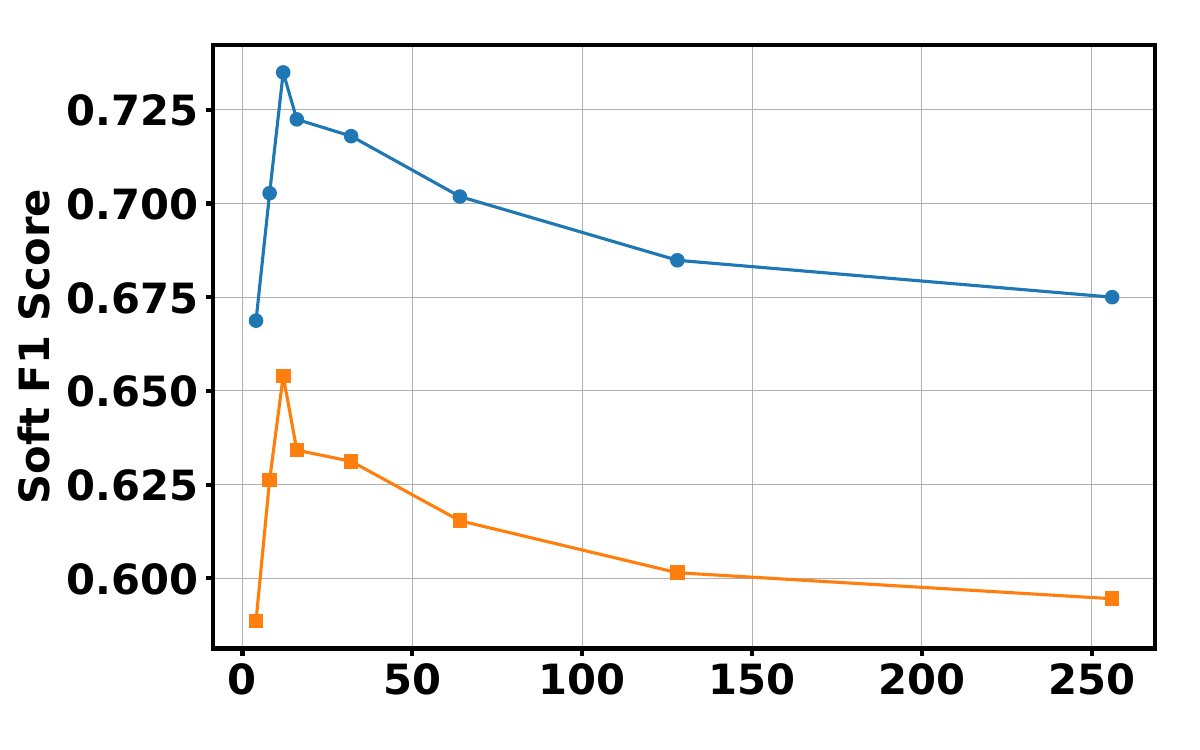}
    \caption{$d$}
    \label{fig:hparam_d}
  \end{subfigure}

  \caption{Hyper-parameter sensitivity on taxonomy quality.}
  \label{fig:hparam}
\end{figure}

\subsection{Ablation Study}
\label{sec:sharper_ablation}
Beyond the configuration-level ablation analysis in Sec.~\ref{sec:main_results}, we conduct two finer-grained ablations that isolate individual design choices: the graph-selection objective (Table~\ref{tab:abl_reliability}) and the budget-aware planner objective (Table~\ref{tab:abl_planner}). (1) \emph{For the graph-selection objective}, coverage-only selection ($\lambda_{\mathrm{rel}}{=}0$) admits noisy partial taxonomies and reduces Final Node F1 by $8.2\%$ relative to the full model, whereas reliability-only selection (removing $f_{\text{cov}}$) over-prunes the candidate graph and reduces it by $5.6\%$; since $f_{\text{cov}}$ optimizes for recall while $f_{\text{rel}}$ optimizes for precision, the two terms are anti-correlated by design, and only their joint optimization recovers the reported $73.5$/$65.4$ Final Node/Edge F1. (2) \emph{For the planner objective}, utility-only planning increases token consumption by $186\%$ (to $1426$k) while delivering a Node F1 that is $9.4\%$ lower than the full model; adding the prior $\Pi^{\text{prior}}$ alone yields a $2.7\%$ relative improvement in Node F1 over utility-only planning but leaves token usage unconstrained, whereas adding the budget constraint alone limits tokens to $542$k and yields a $4.8\%$ relative improvement. Only the complete three-term formulation simultaneously attains high quality ($73.5\%$ Node F1) and tight cost ($498$k tokens), confirming that the five components---$f_{\text{cov}}$, $f_{\text{rel}}$, $U_h$, $\Pi^{\text{prior}}$, and the budget constraint---are jointly indispensable.

\subsection{Diagnostic Studies}
\label{sec:robustness_studies}

\textbf{Robustness of Related Work Detection and Extraction.} We verify two properties of the partial taxonomy extraction pipeline: (i) the coverage of the regex-based detector, and (ii) the accuracy of the LLM-based extractor in recovering the author-curated partial taxonomy from the detected sections. We randomly sample 40 papers from the arXiv benchmark and manually annotate each with its \emph{ground-truth Related Work span} and \emph{ground-truth partial taxonomy}, against which we report detector \emph{coverage} and \emph{false-positive rate}, as well as extractor Node Soft F1 and Edge Soft F1. The regex detector achieves 95\% coverage at a 0\% false-positive rate; the missed 5\% corresponds to 2 papers that contain no explicit Related Work section and are skipped entirely. This confirms that authors overwhelmingly organize related work in a dedicated section with an identifiable heading—a regularity the lightweight detector reliably exploits. The LLM extractor attains 94.7\% Node Soft F1 and 93.5\% Edge Soft F1 on the detected sections, indicating that the recovered partial taxonomies closely align with the manually annotated ground-truth structures and introduce only limited extraction noise into the subsequent taxonomy-integration stage.

\noindent\textbf{Numerical Stability under Deep Hierarchies.} We verify whether box volumes shrink sharply or remain numerically well-behaved as taxonomy depth grows. To this end, we measure the smallest learned box volume and the smallest per-dimension side length across all nodes on two topics: the LLM-based Text-to-SQL taxonomy (5 levels), one of the deepest topics in our benchmark, and an exceptionally deep 11-level taxonomy derived from a broad topic query of \emph{Machine learning}. To isolate the effect of the volume regularization, we compare two variants of the bidirectional mapping model under otherwise identical hyperparameters: \emph{w/o Reg.}, which removes the regularization term, and \emph{w/ Reg.}, our full model that penalizes any box whose volume falls below the threshold $V_{\min}$. 
As a result, without regularization both quantities collapse toward the FP32 underflow limit, and training diverges with NaN gradients on the deeper topic. With regularization, the smallest volumes remain at $1.2 \times 10^{-4}$ and $3.5 \times 10^{-5}$---roughly 33--34 orders of magnitude above the FP32 underflow limit ($\approx 10^{-38}$)---while the smallest side lengths stay at $0.37$ and $0.26$, indicating that inner boxes retain a meaningful spatial extent and the model converges stably in both settings.
\begin{table}[t]
\footnotesize
\begin{minipage}[t]{0.55\linewidth}
\centering
\setlength{\tabcolsep}{3pt}
\caption{Ablation on the reliability--coverage decomposition.}
\label{tab:abl_reliability}
\begin{tabular}{@{}l cc cc@{}}
\toprule
& \multicolumn{2}{c}{Node F1} & \multicolumn{2}{c}{Edge F1} \\
\cmidrule(lr){2-3} \cmidrule(lr){4-5}
Variant & Init & Final & Init & Final \\
\midrule
$f_{\text{cov}}$ only ($\lambda_{\mathrm{rel}}{=}0$)  & 41.7 & 67.5 & 35.7 & 58.9 \\
$f_{\text{rel}}$ only (no coverage)                  & 43.4 & 69.4 & 37.3 & 60.8 \\
\midrule
Full (\method)                                        & \textbf{47.5} & \textbf{73.5} & \textbf{40.8} & \textbf{65.4} \\
\bottomrule
\end{tabular}
\end{minipage}\hfill
\begin{minipage}[t]{0.43\linewidth}
\centering
\setlength{\tabcolsep}{3pt}
\caption{Ablation on the budget-aware planner objective.}
\label{tab:abl_planner}
\begin{tabular}{@{}l c c c@{}}
\toprule
Variant & \begin{tabular}[c]{@{}c@{}}Final\\Node F1\end{tabular} & \begin{tabular}[c]{@{}c@{}}Tokens\\(k)\end{tabular} & \begin{tabular}[c]{@{}c@{}}Cost\\(\$)\end{tabular} \\
\midrule
$U_h$ only                 & 66.6 & 1426 & 4.27 \\
$U_h + \Pi^{\text{prior}}$ & 68.4 & 1403 & 4.20 \\
$U_h + {}$budget           & 69.8 & 542  & 1.62 \\
\midrule
Full (\method)             & \textbf{73.5} & \textbf{498} & \textbf{1.49} \\
\bottomrule
\end{tabular}
\end{minipage}
\end{table}

\noindent\textbf{Predictive Concept Precision.} To verify whether the void volumes identified by our concept emergence model correspond to plausible emergence regions, we measure \emph{Predictive Concept Precision}. For each timestamp transition $t \to t{+}1$, we collect \emph{all} hypothesized concepts used for retrieval queries under each refinement context in $\mathcal{T}_t$, identify \emph{all} newly emerging concepts at $t{+}1$, and compute the fraction of hypothesized concepts that are matched by at least one newly emerging concept under the soft-matching criterion. We report precision rather than recall because ground-truth concept sets at intermediate timestamps are not available in our benchmark, which provides ground-truth taxonomies only at the final timestamp. Averaged over six transitions and 12 benchmark topics, \method achieves 54.2\% precision. The precision is also relatively consistent across the three refinement contexts, with 48.5\%, 55.4\%, and 58.7\% for Specialization, Abstraction, and Expansion, respectively. These results suggest that the geometric gaps provide a useful soft prior for identifying plausible emergence regions. Unrecovered concepts are handled by the budget-aware retrieval module, as hypothesized concepts are used only as a soft prior rather than a hard filter.

\subsection{Case Study: LLM-based Text-to-SQL}
\label{sec:case_llm2sql}

We present a qualitative case study on \emph{LLM-based Text-to-SQL}, which has one of the deepest and most fine-grained ground-truth taxonomies in our benchmark.

\noindent\textbf{Taxonomy Evolution.}
Table~\ref{tab:case_study_taxonomy} tracks the taxonomy's evolution over time.
Initially ($t{=}1$), the structure captures a coarse pipeline with broad subtopics like consistency post-processing and few-shot prompting.
By the middle stage ($t{=}4$), it becomes significantly more granular, separating \emph{Schema Linking} from general pre-processing and refining inference workflows into \emph{Chain-of-Thought} and \emph{Decomposition}, reflecting the 2023--2024 literature.
Finally ($t{=}7$), it captures emerging trends difficult to surface via keyword search, such as \emph{Autonomous Agents}.


\begin{table*}[t]
  \centering
  \footnotesize
  \setlength{\tabcolsep}{4pt}
  \renewcommand{\arraystretch}{1.2}

  \caption{Representative taxonomy evolution stages.}\label{tab:case_study_taxonomy}
  \begin{tabular}{@{}l p{0.52\linewidth} p{0.36\linewidth}@{}}
    \toprule
    \textbf{Stage} & \textbf{Representative New Concepts} & \textbf{Example Evidence Papers} \\
    \midrule
    $t{=}1$ &
    Post-processing $\rightarrow$ Consistency Method; Few-shot $\rightarrow$ Replacing the task description; Model Evaluation &
    BINDER: \emph{Binding language models in symbolic languages} (2022) \\
    \midrule
    $t{=}4$ &
    Pre-processing $\rightarrow$ Schema Linking (Traditional vs.\ LLM-based); Inference $\rightarrow$ Workflow (CoT / Decomposition) &
    ACT-SQL (2023); DIN-SQL (2023) \\
    \midrule
    $t{=}7$ &
    Workflow $\rightarrow$ Autonomous Agents; richer consistency post-processing &
    \emph{Spider 2.0} (2024); R3 (2024); MCS-SQL (2024) \\
    \bottomrule
  \end{tabular}

  \vspace{1.5ex}  

  \caption{Representative paper search results.}\label{tab:case_study_paper_search}
  \begin{tabular}{@{}p{0.30\linewidth} p{0.60\linewidth}@{}}
    \toprule
    \textbf{Concept \& Query} & \textbf{Top-5 Results} \\
    \midrule
    \textbf{Chain-of-Thought (CoT)} \newline
    \emph{``chain-of-thought prompting for text-to-SQL with correction''} &
    \underline{Flat-Ret:} \textbf{\emph{Divide and prompt}} (2023); \emph{Exploring chain of thought style prompting...} (2023); \textbf{ACT-SQL (2023)}; \emph{Text-to-SQL empowered...} (2023) \newline
    \underline{GIST:} \textbf{\emph{Divide and prompt}} (2023); \textbf{ACT-SQL (2023)}; \textbf{CHESS (2024)}; \textbf{CoE-SQL (2024)}; \textbf{OpenSearch-SQL (2025)} \\
    \midrule
    \textbf{Traditional Schema Linking} \newline
    \emph{``schema linking with skeleton retrieval / de-semanticization''} &
    \underline{Flat-Ret:} \textbf{\emph{De-semanticization...}} (2023); \textbf{DBCopilot (2023)}; C3 (2023); \emph{OpenSQL Framework} (2024); \emph{RSL-SQL} (2024) \newline
    \underline{GIST:} \textbf{\emph{De-semanticization...}} (2023); \textbf{CRUSH4SQL (2023)}; \textbf{PURPLE (2024)}; \textbf{SGU-SQL (2024)}; \textbf{OpenSearch-SQL (2025)} \\
    \bottomrule
  \end{tabular}
\end{table*}

\noindent\textbf{Taxonomy-Indexed Search.}
We demonstrate the taxonomy's value as a semantic index in Table~\ref{tab:case_study_paper_search}.
User queries are first aligned to the most relevant concept node, restricting retrieval to that concept's subtree. This minimizes topic drift for overloaded terms (e.g., ``consistency'').
As shown, \method concentrates ranking on ground-truth relevant papers (bolded) more effectively than flat retrieval (\emph{Flat-Ret}), which often retrieves tangentially related work.





\section{Related Work}\label{sec:related_work}

\noindent\textbf{Knowledge Graph Maintenance.}
Maintaining structured knowledge under continuous evolution is a long-standing data management problem~\cite{Weikum21}. Prior efforts can be broadly organized into three categories.
\emph{(i) Incremental knowledge-base construction.} Systems such as DeepDive~\cite{Shin2015DeepDive} and Fonduer~\cite{Wu2018Fonduer} extract entities and relations from newly arriving documents to grow a knowledge base over time.
\emph{(ii) Entity resolution and schema alignment.} A second line consolidates equivalent records across heterogeneous sources through learned matchers and unified resolution frameworks~\cite{li2020deep,tu2023unicorn,huang2023er1,galhotra2018er2}.
\emph{(iii) Continual learning under distribution drift.} Methods such as Elastic Weight Consolidation (EWC)~\cite{kirkpatrick2017catastrophic} adapt model parameters to shifting input distributions, typically guarding against catastrophic forgetting via regularization or rehearsal.
These categories share a common problem setting: the underlying schema is treated as \emph{fixed}, and maintenance proceeds by growing coverage, consolidating duplicates, or adapting parameters within it. \method{} instead targets a setting in which the schema itself---the concept hierarchy---evolves through expansion, abstraction, and specialization. This difference introduces two new challenges: (i) the supervision signal drifts with the hierarchy, so past training signals cannot be naively replayed as in rehearsal-based continual learning; and (ii) maintenance requires coupled growth and structural revision over concepts, rather than insertion under a fixed schema or instance-level record matching.

\noindent\textbf{Taxonomy Construction.} Automated taxonomy construction generally falls into two paradigms. \textit{Entity-driven} methods expand or complete existing structures by predicting concept placement~\cite{Shen2020TaxoExpan,Jiang2022TaxoEnrich,Sun2024TaxoGlimpse,taxonomy_method,Mao2018TaxoRL}, recently incorporating LLMs for low-resource tasks~\cite{Mishra2024FLAME}. Conversely, \textit{document-driven} methods induce hierarchies directly from corpora via clustering~\cite{Zhang2018TaxoGen}, seed-guided growth~\cite{Shen2019HiExpan}, or emerging LLM-based prompting~\cite{Zeng2024ChainOfLayer,kargupta2025taxoadapt,lahiri2025taxoalign,zhu2025context,lan2026agenticscholar}.
These methods treat taxonomy construction as a generation task over fixed, small corpora, leaving LLM-based variants prone to structural hallucination. We instead target \problem in a dynamic, in-the-wild setting and reframe it as a \emph{structural integration} problem: LLMs are restricted to extracting author-curated hierarchical signals from \emph{Related Work} sections, while reliability-aware integration, hypothesized concept prediction from geometric voids, and budget-aware retrieval jointly address the efficiency and cost challenges of streaming paper arrivals.

\noindent\textbf{Geometric Embeddings for Hierarchical Relations.}
To capture the asymmetry of hierarchical relations, researchers have explored non-Euclidean geometries, including 1) hyperbolic models~\cite{Nickel2017Poincare,Ganea2018EntailmentCones,ma2021hyperexpan}, which embed concepts as points in a negatively curved space whose exponentially expanding volume compactly accommodates tree-like structure; 2) order embeddings~\cite{Vendrov2016Order,Lai2017DenotationalProbabilities,Athiwaratkun2018HierarchicalDensity}, which model is-a as a partial order, originally via coordinate-wise dominance and subsequently extended to probabilistic and density-based formulations; and 3) \textit{box embeddings}~\cite{Vilnis2018Box,Li2019SmoothingBox,Dasgupta2020BoxLocal,Abboud2020BoxE}, which represent each concept as an axis-aligned hyperrectangle so that is-a relations correspond to spatial inclusion between boxes. In this work, we choose box embeddings as they provide an explicit geometric prior in which the void volumes between boxes furnish a signal of where new concepts can plausibly emerge. Neither property is available from hyperbolic or order-based representations. Furthermore, extending the existing box-embedding framework, we introduce two components tailored to in-the-wild taxonomy maintenance. First, we propose a novel bidirectional mapping model to bridge the semantic word-embedding space and the geometric box-embedding space. Second, we introduce a geometric concept emergence model that predicts emerging concepts to guide evidence retrieval.

\noindent\textbf{Continual Learning.}
Continual learning aims to update models over a non-stationary data stream without catastrophic forgetting~\cite{kirkpatrick2017catastrophic}, and the literature broadly clusters into three families: \emph{regularization-based} methods that anchor parameters to past optima via importance-weighted penalties (e.g., EWC~\cite{kirkpatrick2017catastrophic}, SI~\cite{zenke2017si}, MAS~\cite{aljundi2018mas}); \emph{replay-based} methods that maintain a small memory of past examples or generative surrogates and rehearse them alongside new data (e.g., iCaRL~\cite{rebuffi2017icarl}, GEM/A-GEM~\cite{lopezpaz2017gem,chaudhry2019agem}, ER~\cite{chaudhry2019er}); and \emph{architecture-based} methods that allocate or isolate sub-networks per task (e.g., Progressive Networks~\cite{rusu2016progressive}, PackNet~\cite{mallya2018packnet}). These approaches share a common operating assumption: the \emph{label space} or task structure is externally specified, and the goal is to preserve model behavior on previously seen tasks while accommodating new ones~\cite{delange2021survey}. \method{} departs from this setting in two ways. First, the supervision signal in our problem is itself drifting: the taxonomy that defines parent--child relations grows, abstracts, and specializes over time, so there is no fixed task identity to anchor regularization or rehearsal against. Second, our novelty-aware coreset selection is principled as an importance-reweighted approximation of the historical empirical risk (Sec.~\ref{sec:coreset_update}), and it explicitly trades off ERM fidelity against a novelty score that quantifies complementarity with the incoming batch—yielding a memory that is both \emph{representative} of historical structural signals and \emph{non-redundant} with newly extracted ones, a property that standard reservoir or random-sampling replay buffers do not provide.

\noindent\textbf{Data Integration.}
Data integration is a fundamental problem that arises across diverse data modalities, including tabular data~\cite{christen2012datamatching, Konda2016Magellan, li2020deep, tu2023unicorn, ji2025table, lee2026salto}, graph data~\cite{Suchanek2008YAGO, Auer2007DBpedia, Weikum21, huang2023er1}, and text corpora~\cite{Shin2015DeepDive, Wu2018Fonduer}. Our work is most closely related to the integration of graph-structured data, as it dynamically and incrementally consolidates multiple partial taxonomies into a unified, evolving taxonomy. For \emph{graph data}, integration centers on aligning and fusing heterogeneous graph-structured sources. Representative efforts include large-scale knowledge graph construction that merges Wikipedia, WordNet, and other sources into unified ontologies~\cite{Suchanek2008YAGO, Auer2007DBpedia, Weikum21}, as well as entity alignment across knowledge graphs, where embedding-based and active learning methods jointly align entities and schemata~\cite{huang2023er1}. A related direction leverages integrated graph structures for downstream retrieval and reasoning, such as retrieval-augmented generation over textual graphs~\cite{He2024GRetriever} and hierarchical, query-focused summarization over entity-relation graphs~\cite{edge2024local}. Unlike conventional graph integration methods that align and fuse a fixed collection of pre-existing graphs, GIST continuously incorporates partial taxonomies extracted from newly arriving documents and revises the global concept hierarchy as the underlying scholarly repository evolves.

\section{Conclusion}
\label{sec:conclusion}

We address \problem over evolving scholarly repositories. We propose \method, which leverages human-curated partial taxonomies as reliable signals and integrates them via geometric containment relationships, together with a probabilistic hypothesized concept generation module and a cost-effective evidence retrieval module. Experiments on the real-world arXiv corpus (ar5iv) across 12 topics show that \method achieves strong taxonomy quality with substantially lower runtime and monetary cost than LLM-driven baselines.

\begin{acks}
Zhifeng Bao is supported in part by ARC DP240101211 and FT240100832. Junhao Gan is supported in part by ARC DP230102908.
\end{acks}


\bibliographystyle{ACM-Reference-Format}
\bibliography{reference}

\clearpage

\section*{Appendix}

\appendix
\section{Proofs}

\subsection{Proof of Submodularity for Eq.~\eqref{eq:cover_rel}}~\label{sec:proof_submodularity_f1}
\begin{lemma}[Submodularity of the coverage term]\label{lem:submod_cov}
The function $f_{\text{cov}}:2^{\mathcal{E}_{\text{cand}}}\to\mathbb{R}_{\ge 0}$ is monotone and submodular.
\end{lemma}

\begin{proof}
Let $g(\mathcal{D}) := \bigl|\bigcup_{e\in\mathcal{D}} S_e\bigr|$. This is a standard \emph{coverage} (set-union) function and is monotone submodular~\cite{krause2014submodular}. Since $f_{\text{cov}}(\mathcal{D}) = g(\mathcal{D})/|\mathcal{V}_{\text{cand}}|$ is just a positive scaling of $g$, $f_{\text{cov}}$ remains monotone and submodular.
\end{proof}

\begin{lemma}[Modularity of the reliability term]\label{lem:mod_rel}
The function $f_{\text{rel}}:2^{\mathcal{E}_{\text{cand}}}\to\mathbb{R}$ is modular (hence submodular). Moreover, since $\tilde{w}_e\in[0,1]$, it is monotone.
\end{lemma}

\begin{proof}
$f_{\text{rel}}(\mathcal{D})=\frac{1}{K}\sum_{e\in\mathcal{D}}\tilde{w}_e$ is a weighted sum over selected elements, hence modular (its marginal gain for adding $e$ is always $\tilde{w}_e/K$, independent of $\mathcal{D}$). Since $\tilde{w}_e\ge 0$, it is also monotone.
\end{proof}

\begin{proposition}[Submodularity of $F$]\label{prop:submod_F}
The objective $F(\mathcal{D}) = f_{\text{cov}}(\mathcal{D}) + \lambda_{\text{rel}} f_{\text{rel}}(\mathcal{D})$ is submodular for any $\lambda_{\text{rel}}\in\mathbb{R}$. If $\lambda_{\text{rel}}\ge 0$, then $F$ is also monotone and non-negative.
\end{proposition}

\begin{proof}
Coverage functions are submodular (Lemma~\ref{lem:submod_cov}) and modular functions are submodular (Lemma~\ref{lem:mod_rel}). Submodularity is closed under addition and scalar multiplication~\cite{krause2014submodular}, so $F$ is submodular for any $\lambda_{\text{rel}}$. If $\lambda_{\text{rel}}\ge 0$, both terms are monotone and non-negative, hence so is $F$.
\end{proof}


\subsection{Proof for Proposition~\ref{prop:coreset_sampling}.}\label{sec:proof_coreset_sampling}

\begin{proof}
Let $\mathcal{D}_{\mathrm{old}}$ be the domain of edges, and write $\Pi(e)$ as shorthand for
$\Pi_{\text{core}}(e)$.
Eq.~\ref{eq:nov-kl} is
\[
\max_{\Pi \in \Delta(\mathcal{D}_{\mathrm{old}})}
\;-\mathrm{KL}\!\left(\Pi \,\|\, \Pi_{\mathrm{ERM}}\right)
\;+\;
\lambda_{\mathrm{nov}}\,\mathbb{E}_{e\sim \Pi}\!\left[n(e)\right].
\]
Expanding the two terms yields the equivalent objective
\[
\max_{\Pi \in \Delta(\mathcal{D}_{\mathrm{old}})}
\;\sum_{e}\Pi(e)\log\frac{\Pi_{\mathrm{ERM}}(e)}{\Pi(e)}
\;+\;
\lambda_{\mathrm{nov}}\sum_{e}\Pi(e)\,n(e).
\]
Since $\mathrm{KL}(\cdot\|\Pi_{\mathrm{ERM}})$ is strictly convex over the simplex
(on the support of $\Pi_{\mathrm{ERM}}$), the objective is strictly concave,
hence the maximizer is unique.

Introduce a Lagrange multiplier $\alpha\in\mathbb{R}$ for $\sum_e \Pi(e)=1$.
The Lagrangian is
\[
\mathcal{L}(\Pi,\alpha)
=
\sum_{e}\Pi(e)\log\frac{\Pi_{\mathrm{ERM}}(e)}{\Pi(e)}
+
\lambda_{\mathrm{nov}}\sum_{e}\Pi(e)\,n(e)
+
\alpha\Big(\sum_{e}\Pi(e)-1\Big).
\]
For any $e$ with $\Pi_{\mathrm{ERM}}(e)>0$, optimality implies $\Pi^*(e)>0$,
and the stationarity condition gives
\[
\frac{\partial \mathcal{L}}{\partial \Pi(e)}
=
\log \Pi_{\mathrm{ERM}}(e) - \log \Pi(e) - 1
+\lambda_{\mathrm{nov}}\,n(e)+\alpha
=0.
\]
Rearranging,
\[
\log \Pi(e)
=
\log \Pi_{\mathrm{ERM}}(e)
+\lambda_{\mathrm{nov}}\,n(e)
+(\alpha-1),
\]
and exponentiating yields
\[
\Pi^*(e)
=
\Pi_{\mathrm{ERM}}(e)\exp\!\big(\lambda_{\mathrm{nov}}\,n(e)\big)\cdot \exp(\alpha-1).
\]
Enforcing $\sum_e \Pi^*(e)=1$ gives the normalizer
\[
\exp(\alpha-1)
=
\Bigg[\sum_{e'} \Pi_{\mathrm{ERM}}(e')\exp\!\big(\lambda_{\mathrm{nov}}\,n(e')\big)\Bigg]^{-1}.
\]
Therefore,
\[
\Pi_{\text{core}}^{\ast}(e)
=
\frac{\Pi_{\text{ERM}}(e)\exp\!\big\{\lambda_{\mathrm{nov}}\,n(e)\big\}}
{\sum_{e'} \Pi_{\text{ERM}}(e')\exp\!\big\{\lambda_{\mathrm{nov}}\,n(e')\big\}}.
\]
For any $e$ with $\Pi_{\mathrm{ERM}}(e)=0$, any feasible $\Pi(e)>0$ makes
$\mathrm{KL}(\Pi\|\Pi_{\mathrm{ERM}})=+\infty$, so necessarily $\Pi^*(e)=0$,
which is also consistent with the expression above. This completes the proof.
\end{proof}

\subsection{Proof for Proposition~\ref{prop:coreset_weights}}\label{sec:proof_coreset_weights}
\begin{proof}
(i) Unbiasedness is standard importance sampling. Since $e\sim\Pi_{\mathrm{core}}$,
\begin{equation}
\begin{aligned}
\mathbb{E}\!\left[\frac{\Pi_{\mathrm{ERM}}(e)}{\Pi_{\mathrm{core}}(e)}\ell(e;\theta)\right]
&=\sum_e \Pi_{\mathrm{core}}(e)\frac{\Pi_{\mathrm{ERM}}(e)}{\Pi_{\mathrm{core}}(e)}\ell(e;\theta)\\
&=\sum_e \Pi_{\mathrm{ERM}}(e)\ell(e;\theta)=\mathcal{L}_{\mathrm{old}}(\theta),
\end{aligned}
\end{equation}
and linearity of expectation yields
$\mathbb{E}[\widehat{\mathcal{L}}_{\mathrm{old}}(\theta)]=\mathcal{L}_{\mathrm{old}}(\theta)$.

(ii) The self-normalized estimator can be written as
\[
\widetilde{\mathcal{L}}_{\mathrm{old}}(\theta)
=
\frac{\frac{1}{m}\sum_{i=1}^m X_i}{\frac{1}{m}\sum_{i=1}^m Y_i},
\quad
X_i:=\hat{w}_i\,\ell(e_i;\theta),\;\; Y_i:=\hat{w}_i.
\]
Under the stated integrability condition, the strong law of large numbers gives
$\frac{1}{m}\sum_i X_i \to \mathbb{E}[X_1]=\mathcal{L}_{\mathrm{old}}(\theta)$ and
$\frac{1}{m}\sum_i Y_i \to \mathbb{E}[Y_1]=\sum_e \Pi_{\mathrm{ERM}}(e)=1$ almost surely.
Thus the ratio converges to $\mathcal{L}_{\mathrm{old}}(\theta)$ almost surely.
This is the standard consistency result for self-normalized importance sampling;
see, e.g., \cite{robert2004montecarlo}.
\end{proof}

\subsection{Proof of proposition~\ref{pro:exploration_policy}.}
\label{sec:proof_of_constrained_problem}

\begin{proof}
\noindent\textbf{Step 1: Existence and uniqueness.}
Let $S:=\{h:\Pi_{t-1,h}^{\mathrm{prior}}>0\}$ be the support of the prior.
Any feasible $\Pi_t$ with $\Pi_{t,h}>0$ for some $h\notin S$ yields
$\mathrm{KL}(\Pi_t\|\Pi_{t-1}^{\mathrm{prior}})=+\infty$ and is therefore never optimal.
Hence it suffices to optimize over
$\{\Pi_t\in\Delta^{H_{t-1}}:\Pi_{t,h}=0\ \forall h\notin S\}$.
On this domain, $\mathrm{KL}(\cdot\|\Pi_{t-1}^{\mathrm{prior}})$ is strictly convex in $\Pi_t$ (on the relative interior),
so the objective is strictly concave; the feasible set is closed and bounded (and nonempty under budget feasibility),
thus an optimum exists and is unique.

\noindent\textbf{Step 2: KKT conditions and closed form.}
Introduce Lagrange multipliers $\alpha\in\mathbb{R}$ for $\sum_h\Pi_{t,h}=1$ and $\nu\ge 0$ for the budget constraint.
The Lagrangian is
\[
\begin{aligned}
\mathcal{L}(\Pi_t,\alpha,\nu)
&=\sum_{h}\Pi_{t,h}U_h
-\lambda_{\mathrm{plan}}\sum_{h}\Pi_{t,h}\log\frac{\Pi_{t,h}}{\Pi_{t-1,h}^{\mathrm{prior}}} \\
&\quad-\alpha\Big(\sum_h\Pi_{t,h}-1\Big)
-\nu\Big(M_t\sum_h\Pi_{t,h}\bar{\ell}_h-K_t\Big).
\end{aligned}
\]
For any $h\in S$, the optimal solution satisfies $\Pi_{t,h}^*>0$, and the stationarity condition
$\partial\mathcal{L}/\partial\Pi_{t,h}=0$ gives
\[
U_h
-\lambda_{\mathrm{plan}}\Big(\log\frac{\Pi_{t,h}^*}{\Pi_{t-1,h}^{\mathrm{prior}}}+1\Big)
-\alpha
-\nu M_t\bar{\ell}_h
=0.
\]
Rearranging,
\[
\begin{aligned}
\log\frac{\Pi_{t,h}^*}{\Pi_{t-1,h}^{\mathrm{prior}}}
&=\frac{U_h-\nu M_t\bar{\ell}_h-\alpha-\lambda_{\mathrm{plan}}}{\lambda_{\mathrm{plan}}},\\
\Pi_{t,h}^*
&=\Pi_{t-1,h}^{\mathrm{prior}}
\exp\!\Big(\frac{U_h-\nu M_t\bar{\ell}_h}{\lambda_{\mathrm{plan}}}\Big)\cdot C,
\end{aligned}
\]
where $C=\exp\!\big(-(\alpha+\lambda_{\mathrm{plan}})/\lambda_{\mathrm{plan}}\big)$ is constant in $h$.
Enforcing $\sum_h\Pi_{t,h}^*=1$ yields
\[
C
=\Bigg[\sum_{h'=1}^{H_{t-1}}
\Pi_{t-1,h'}^{\mathrm{prior}}
\exp\!\Big(\frac{U_{h'}-\nu M_t\bar{\ell}_{h'}}{\lambda_{\mathrm{plan}}}\Big)\Bigg]^{-1}.
\]
Define $\lambda:=\nu M_t\ge 0$ and
\[
Z(\lambda)
=\sum_{h'=1}^{H_{t-1}}
\Pi_{t-1,h'}^{\mathrm{prior}}
\exp\!\Big(\frac{U_{h'}-\lambda \bar{\ell}_{h'}}{\lambda_{\mathrm{plan}}}\Big).
\]
Then, for all $h$ (including $h\notin S$, where $\Pi_{t-1,h}^{\mathrm{prior}}=0$ implies $\Pi_{t,h}^*=0$),
\[
\Pi_{t,h}^{*}
=\frac{\Pi_{t-1,h}^{\mathrm{prior}}
\exp\!\big((U_h-\lambda \bar{\ell}_h)/\lambda_{\mathrm{plan}}\big)}{Z(\lambda)}.
\]
This matches the claimed Gibbs form for some optimal dual multiplier $\lambda^*\ge 0$.

\noindent\textbf{Step 3: Determining $\lambda^*$ and bisection.}
The remaining KKT conditions are primal feasibility and complementary slackness:
\[
\begin{aligned}
M_t\sum_{h=1}^{H_{t-1}}\Pi_{t,h}^{*}\bar{\ell}_h &\le K_t,\\
\nu\Big(M_t\sum_{h=1}^{H_{t-1}}\Pi_{t,h}^{*}\bar{\ell}_h-K_t\Big) &= 0,
\qquad \nu\ge 0.
\end{aligned}
\]
Equivalently, either the budget is inactive and $\lambda^*=\nu M_t=0$, or it is active and
\[
M_t\sum_{h=1}^{H_{t-1}}\Pi_{t,h}(\lambda^*)\,\bar{\ell}_h = K_t,
\]
where $\Pi_{t,h}(\lambda)$ denotes the Gibbs distribution above.

Let
\[
c(\lambda):=\sum_{h=1}^{H_{t-1}}\Pi_{t,h}(\lambda)\bar{\ell}_h.
\]
A standard exponential-family calculation yields that $c(\lambda)$ is continuous and nonincreasing:
\[
\begin{aligned}
\frac{d}{d\lambda}c(\lambda)
&= -\frac{1}{\lambda_{\mathrm{plan}}}\,
\operatorname{Var}_{\Pi_t(\lambda)}(\bar{\ell}_h)
\;\le\;0.
\end{aligned}
\]
Therefore $g(\lambda):=M_t c(\lambda)-K_t$ is continuous and nonincreasing.
If $g(0)\le 0$, then $\lambda^*=0$ satisfies the KKT conditions.
Otherwise $g(0)>0$ and, under feasibility (e.g., $M_t\min_h\bar{\ell}_h\le K_t$),
we have $\lim_{\lambda\to\infty} g(\lambda)\le 0$; thus there exists $\lambda^*>0$ with $g(\lambda^*)=0$.
Monotonicity allows finding $\lambda^*$ via one-dimensional bisection, and the constraint holds with equality when active.
\end{proof}

\subsection{Proof of Submodularity for Eq.~\eqref{eq:evidence_retrieval}}~\label{sec:proof_submodularity_f2}
\begin{proof}
We prove that $F$ is both monotone and submodular.

\noindent\textbf{Monotonicity.}
For any $A\subseteq B$, for every $\mathbf{x}$, we have
$$\max_{d\in A} a(d,\mathbf{x}) \le \max_{d\in B} a(d,\mathbf{x}),$$  hence
$f_{\mathbf{x}}(A)\le f_{\mathbf{x}}(B)$ and therefore $F(A)\le F(B)$.

\noindent\textbf{Submodularity.}
Fix any $A\subseteq B$ and any $e\notin B$.
For each $\mathbf{x}$, let
\[
m_A(\mathbf{x})=\max_{d\in A} a(d,\mathbf{x}),
\qquad
m_B(\mathbf{x})=\max_{d\in B} a(d,\mathbf{x}),
\]
so that $m_A(\mathbf{x})\le m_B(\mathbf{x})$.
The marginal gain of adding $e$ to a set $S$ for term $\mathbf{x}$ is
$$
\begin{aligned}
&f_{\mathbf{x}}(S\cup\{e\})-f_{\mathbf{x}}(S)
=
\max\{m_S(\mathbf{x}),\\
&a(e,\mathbf{x})\}-m_S(\mathbf{x})
=
\max\{0,\,a(e,\mathbf{x})-m_S(\mathbf{x})\}.
\end{aligned}
$$
Since $m_A(\mathbf{x})\le m_B(\mathbf{x})$, we have
\[
\max\{0,\,a(e,\mathbf{x})-m_A(\mathbf{x})\}
\;\ge\;
\max\{0,\,a(e,\mathbf{x})-m_B(\mathbf{x})\},
\]
i.e., the marginal gain under $A$ is at least that under $B$ for every $\mathbf{x}$.
Summing over $\mathbf{x}\in\mathcal{X}_t$ yields
\[
F(A\cup\{e\})-F(A)\;\ge\;F(B\cup\{e\})-F(B),
\]
so $F$ is submodular. Finally, a sum of submodular functions is submodular, completing the proof.
\end{proof}

\section{Details for LLM-guided Partial Taxonomy Extraction}\label{sec:details_for_LLM_extractor}
\subsection{Regex Matching for Extracting Related Work Sections}
To acquire author-curated knowledge about how a paper organizes prior studies around a given research topic, we extract \emph{Related Work} (and semantically equivalent) sections using a lightweight rule-based pipeline based on regular-expression matching.
In the arXiv dataset, the LaTeX2HTML pages structure the document into section-level HTML blocks annotated with the CSS class \texttt{ltx\_section}.
We first segment each paper into \texttt{ltx\_section} units, parse the corresponding section titles, and normalize them (lowercasing and stripping punctuation and numbering).
We then match titles against a case-insensitive library of candidate patterns that cover common heading variants, including \emph{Related Work}, \emph{Background}, \emph{Literature Review}, \emph{Prior/Previous Work}, and compound forms such as \emph{Background and Related Work}.
If multiple sections match, we prioritize the earliest highest-level match (to avoid appendices) and include its immediate subsections when they are topically consistent.
Finally, we convert the matched section(s) into plain text for downstream processing, removing non-prose artifacts (e.g., figure/table captions and equations) while preserving in-text citation markers to retain bibliographic cues.

\subsection{Prompt Design}
Given an extracted related-work text span and a target research topic, we use a structured prompt to instruct an LLM to:
(1) identify the categorization scheme implicitly or explicitly used by the authors to organize prior work, and
(2) output these categories as a machine-readable taxonomy in JSON.
To improve reliability and ease of parsing, we constrain the model to a fixed JSON schema (a rooted tree with optional intermediate levels) and explicitly prohibit free-form commentary.
We further require topical filtering: the model must keep only categories that are directly relevant to the given research topic, and discard generic or off-topic buckets (e.g., ``other applications'') unless they contain topic-specific subcategories.
To reduce hallucination, we instruct the model to ground each extracted node in the provided text (e.g., by attaching a short supporting span) and to return an empty taxonomy when the section does not provide a discernible categorization.

\section{Details for Experiment Setting}\label{sec:details_for_experiment_setting}
\subsection{Survey Papers}
The survey papers used in our benchmark are listed in Table~\ref{tab:survey_paper}. They are primarily drawn from high-impact venues, including ACM Computing Surveys (CSUR), IEEE Transactions on Knowledge and Data Engineering (TKDE), and ACM Transactions on Information Systems (TOIS).
\begin{table*}[t]
\centering
\caption{Details of the benchmark survey papers}\label{tab:survey_paper}
\small
\setlength{\tabcolsep}{6pt}
\begin{tabularx}{\linewidth}{X l c}
\toprule
\textbf{Title} & \textbf{Venue} & \textbf{Year} \\
\midrule
Extended Reality (XR) Toward Building Immersive Solutions: The Key to Unlocking Industry 4.0
& \emph{ACM Computing Surveys} & 2024 \\

Self-tuning Database Systems: A Systematic Literature Review of Automatic Database Schema Design and Tuning
& \emph{ACM Computing Surveys} & 2024 \\

HTAP Databases: A Survey
& \emph{IEEE Transactions on Knowledge and Data Engineering} & 2024 \\

A Survey on Employing Large Language Models for Text-to-SQL Tasks
& \emph{ACM Computing Surveys} & 2026 \\

Graph Retrieval-Augmented Generation: A Survey
& \emph{ACM Transaction on Information Systems} & 2025 \\

A Survey of Distributed Graph Algorithms on Massive Graphs
& \emph{ACM Computing Surveys} & 2025 \\

Evolving Paradigms in Automated Program Repair: Taxonomy, Challenges, and Opportunities
& \emph{ACM Computing Surveys} & 2025 \\

Artificial Intelligence Applied to Software Testing: A Tertiary Study
& \emph{ACM Computing Surveys} & 2024 \\

Machine Unlearning: A Survey
& \emph{ACM Computing Surveys} & 2024 \\

Fairness in Machine Learning: A Survey
& \emph{ACM Computing Surveys} & 2024 \\

A Systematic Literature Review on Multi-Robot Task Allocation
& \emph{ACM Computing Surveys} & 2025 \\

Multimodal Recommender Systems: A Survey
& \emph{ACM Computing Surveys} & 2025 \\
\bottomrule
\end{tabularx}
\end{table*}

\noindent\textbf{Budgeted Retrieval and Learned Query Planning.}
A long line of data-management work has studied how to retrieve or process data under explicit resource constraints, including budget-constrained top-$k$ and cascade ranking that progressively applies more expensive scorers to shrinking candidate sets~\cite{wang2011cascade,matsubara2020cascade}, learned query optimizers and cost models that estimate per-plan latency to choose execution strategies (e.g., Neo~\cite{marcus2019neo}, Bao~\cite{marcus2021bao}, learned cardinality estimators~\cite{kipf2019mscn,yang2019deepUnsup}), and budgeted active retrieval that selects which documents or labels to acquire under a fixed query budget~\cite{settles2009active}. The unifying assumption across these works is that the \emph{query} (or the workload distribution) is given and fixed, and the optimization target is to minimize cost or maximize relevance for that query. \method{} addresses a fundamentally different regime in which the queries themselves do not yet exist: at each iteration, we must \emph{synthesize} candidate exploration directions from the current taxonomy via hypothesized concept boxes, score how strongly the incoming paper batch can support each direction (utility anchoring), and only then allocate token and API budget across these self-generated queries (Sec.~\ref{sec:planning}). The resulting KL-regularized allocation admits a closed-form softmax solution over $(U_h - \lambda^* \bar{\ell}_h)$, which to our knowledge has no direct counterpart in fixed-query budgeted retrieval or learned query planning. In this sense, \method{} extends budget-aware retrieval from query \emph{execution} to query \emph{generation} under an evolving structural target.

\end{document}